\documentclass[11pt, a4paper]{article}
\def\dj{d\kern-0.4em\char"16\kern-0.1em}

\usepackage{a4wide}

\usepackage{ragged2e}
\usepackage[utf8]{inputenc}
\usepackage[T1]{fontenc}
\usepackage[english]{babel}
\usepackage{subcaption}
\usepackage{amssymb, amsmath, latexsym, amsthm, enumerate, amsfonts}
\usepackage[labelsep=period]{caption}

\usepackage[font=small]{caption}
\usepackage{mwe}

\usepackage[pdftex, pdfstartview=FitH]{hyperref}
\usepackage{graphicx}

\numberwithin{equation}{section}

 \newtheorem{thm}{Theorem}[section]

 \newtheorem{cl}[thm]{Claim}
 \newtheorem{obs}[thm]{Observation}
 
 \theoremstyle{definition}

 \theoremstyle{remark}

\providecommand{\keywords}[1]
{
  \small	
  \textbf{Keywords:} #1
}

\begin{document}

\title{On strong avoiding games}


\author{Milo\v s Stojakovi\' c\thanks{Department of Mathematics and Informatics, Faculty of Sciences, University of Novi Sad, 21000 Novi Sad, Serbia, e-mail: \href{mailto:milosst@dmi.uns.ac.rs}{milosst@dmi.uns.ac.rs}. Partly supported by Ministry of Education, Science and Technological Development of the Republic of Serbia (Grant No.~451-03-68/2022-14/200125). Partly supported by Provincial Secretariat for Higher Education and Scientific Research, Province of Vojvodina (Grant No.~142-451-2686/2021).} \and
Jelena Stratijev\footnote{Department of Fundamental Sciences, Faculty of Technical Sciences, University of Novi Sad, 21000 Novi Sad, Serbia, e-mail: \href{mailto:jelenaknezevic@uns.ac.rs}{jelenaknezevic@uns.ac.rs}. Partly supported by Ministry of Education, Science and Technological Development of the Republic of Serbia (Grant No.~451-03-68/2022-14/200156).} \footnote{Corresponding author.}}

\date{}

\maketitle

\begin{abstract}
Given an increasing graph property $\cal F$, the strong Avoider-Avoider $\cal F$ game is played on the edge set of a complete graph. Two players, Red and Blue, take turns in claiming previously unclaimed edges with Red going first, and the player whose graph possesses $\cal F$ first loses the game. If the property $\cal F$ is ``containing a fixed graph $H$'', we refer to the game as the $H$ game.

We prove that Blue has a winning strategy in two strong Avoider-Avoider games, $P_4$ game and ${\cal CC}_{>3}$ game, where ${\cal CC}_{>3}$ is the property of having at least one connected component on more than three vertices.

We also study a variant, the strong CAvoider-CAvoider games, with additional requirement that the graph of each of the players must stay connected throughout the game. We prove that Blue has a winning strategy in the strong CAvoider-CAvoider games $S_3$ and $P_4$, as well as in the $Cycle$ game, where the players aim at avoiding all cycles.

\end{abstract}



\section{Introduction}
A positional game is a pair $(X,\mathcal{F})$, where $X$ is a finite set called a \emph{board}, and $\mathcal{F}$ is the family of \emph{target sets}. The game is played by two players who alternately claim previously unclaimed elements of $X$ until all the elements of the board are claimed. Our interest lies with games whose board is the edge set of the complete graph $K_n$. Numerous details about positional games, and particularly positional games played on graphs, can be found in books~\cite{TicTacToe} and~\cite{BookStojakovic}.

\medskip

When it comes to the rules for determining the game winner in a positional game, there are several variants. In the \emph{strong Maker-Maker game} $(X,\mathcal{F})$, two players called Red and Blue take turns in claiming previously unclaimed elements of $X$, with Red going first. The player who first fully occupies some $F \in \mathcal{F}$ is the winner. If neither of the players wins and all the elements of the board are claimed, the game is declared a \emph{draw}.

The most notable example of this class of positional games is the widely popular game of Tic-Tac-Toe. Generally speaking, determining the outcome of such games proves to be quite challenging, and there are hardly any general tools at disposal. One such tool is the strategy stealing argument which we can use to show that Red can guarantee at least a draw in any game. Ramsey property of the board with the target sets readily ensures that the draw is impossible. Almost at the end of the list are pairing strategies, which we can use to show that Blue can guarantee a draw.

Hence, it is not too surprising that so few results on strong Maker-Maker games on graphs can be found in the literature. Ferber and Hefetz~\cite{Hefetzstrongfast} proved that playing on the edge set of $K_n$, for sufficiently large $n$, Red can win the perfect matching and Hamilton cycle game, and in~\cite{Hefetzstrongkconnectivity} the same authors proved that, for sufficiently large $n$ and every positive integer $k$, the first player can win $k$-vertex-connectivity game. Both papers rely on fast winning strategies for weak games.

\medskip

The \emph{Strong Avoider-Avoider game} $(X,\mathcal{F})$ is again played by Red and Blue, but now the player who first fully occupies some $F \in \mathcal{F}$ \emph{loses the game}. The first such game, widely known as \emph{Sim}, is introduced in 1961 by Simmons~\cite{Simmons69}. The board of Sim is the edge set of $K_6$, and a player who first claims a triangle loses. Even though it is immediate that draw is impossible, and the board is reasonably small -- it has just fifteen edges, analyzing it is challenging, and the proof that Blue wins is performed with the help of a computer.
In~\cite{Slany99} Slany gave a methodical study of the hardness of determining the winner for several games similar to Sim. Further, Mead, Rosa and Huang in~\cite{SIM74} gave an explicit winning strategy for Blue in Sim, and recently in \cite{Sim20}~Wrzos-Kaminska gave a simple human-playable winning strategy. Other variants of strong Avoider-Avoider games were studied by Harary in~\cite{Harary}, who introduced several finite games on graphs on up to six vertices.

At first sight it may seem that in strong Avoider-Avoider games, in contrast to the strong Maker-Maker games, Blue always has an upper edge, and Red as the first player cannot expect to win under optimal play? This, however, turns out not to be true! For example, in $d$-dimensional Tic-Tac-Toe game $n^d$ (see~\cite{TicTacToe} for details), where $n$ is odd, Red has an explicit drawing strategy:
In his first move he chooses the central element, denote it by $C$. After that, whenever Blue chooses an element $P$ Red chooses $P'$ that is symmetrical with respect to $C$. If we suppose for a contradiction that Red loses, i.e.~that his graph has a red line $L$ (note that it is not possible that $C$ belongs to $L$), then $L'$, its mirror image over the cube's center, is a blue line and has been fully occupied before $L$, a contradiction.
Now, as game $3^3$ cannot end in a draw~\cite{TicTacToe}, we can conclude that it is a Red's win.
In~\cite{TransitiveAvoidance} Johnson, Leader and Walters proved that there are transitive games that are a Red's win, for all board sizes which are not a prime, or a power of 2.

\medskip

In contrast to the strong positional games where the two players compete for achieving the same objective, the \emph{weak games} are asymmetrical -- the first player is given a goal while the second one just tries to prevent the first player from achieving his goal. In \emph{Maker-Breaker positional game} $(X,\mathcal{F})$, two players are called Maker and Breaker. Maker wins the game if by the end of the game he claims all elements of some $F \in \mathcal{F}$, otherwise Breaker wins the game. There is a number of results on Maker-Breaker games on graphs obtained in recent decades, an overview together with further literature can be found in~\cite{BookStojakovic}.

\medskip

\emph{Avoider-Enforcer games} are the \emph{mis\` ere} version of Maker-Breaker games, with two players Avoider and Enforcer. Enforcer wins the game $(X,\mathcal{F})$ if, by the end of the game, Avoider claimed all elements of some $F \in \mathcal{F}$, otherwise Avoider wins. Except these rules, there is another variant, the so-called monotone rules introduced in \cite{StojakovicTheRulesOfTheGame}, where each of the players is allowed to claim more than one element of the board per move. Again, a number of results on Avoider-Enforcer games on graphs can be found in the literature, see~\cite{BookStojakovic} for an overview.

\paragraph{Our results.}
We will take a closer look at strong Avoider-Avoider games. Even though their definition is natural and many questions about them have been asked, very few of them have been answered. To offer some intuition behind this phenomenon, we should keep in mind that the players in strong games have the same goal and the only thing that makes a difference is who goes first, we call this the ``half-a-move advantage''. Informally speaking, depending on the structure of the board there are different ways things can play out, but that half-a-move eventually decides the game. So the player that can win should propagate his (in most cases, comparatively small) advantage from beginning to the end of the game,
knowing that one wrong move may take the edge away from him.
In contrast to this, in weak games we have more freedom when designing a winning strategy, as players have different goals. This further allows the introduction of bias, first time introduced in~\cite{ChE}, which gives us more room to spare. Hence, in most of the weak games studied in the literature we are not that close to the breaking point at which a player stops winning and starts losing.

In this paper we are interested in Strong Avoider-Avoider games played on the edges of the complete graph $K_n$. Not much is known about these games, while there are many open problems.
In~\cite{HararySlanyVerbitsky} was shown that Blue has a winning strategy in the $P_3$ game, where the forbidden graph is the path with just two edges.
Recently, Beker~\cite{BekerStarAA} generalized this result to all stars, proving that for each fixed $k$ the Strong Avoider-Avoider star $S_{k+1}$ game is a win for the second player for all $n$ sufficiently large. The proof is performed by actually building rather than avoiding -- showing that Blue can build a $S_{k+1}$-free graph of maximum size fast, without wasting any moves, thus automatically securing a win.
Finally, Malekshahian~\cite{Malekshahian} studied the possibility of Blue's win in the triangle game with assumption that the game starts on several special mid-game positions, without any definite implications on the outcome of the triangle game itself. Hence, the only non-trivial Strong Avoider-Avoider game played on $E(K_n)$ for which the outcome is previously known is the star game.

\medskip

We use the abbreviation ${\cal CC}_{>3}$ for the collection of inclusion-minimal connected graphs on more than three vertices and $P_4$ represents a path on four vertices. Our goal is to determine the outcome for the $P_4$ game and the ${\cal CC}_{>3}$ game.

\begin{thm}\label{TH001}
Blue has a winning strategy in the Strong Avoider-Avoider $P_4$ game, played on $K_n$, where $n \geq 8$.
\end{thm}

In the following theorem we consider the game where a player loses the game as soon as he creates a connected component on more than three vertices.

\begin{thm}\label{TH002}
Blue has a winning strategy in the Strong Avoider-Avoider ${\cal CC}_{>3}$ game, played on $K_n$, where $n \geq 5$.
\end{thm}

Let $R(F)$ be diagonal Ramsey number, so every 2-coloring of edges of a complete graph on at least $R(F)$ vertices gives a monochromatic $F$. If $n \geq R(F)$ we know that the strong Avoider-Avoider F game on $E(K_n)$ cannot end in a draw. For both the $P_4$ game and the ${\cal CC}_{>3}$ game this readily implies that there is no draw for $n \geq 5$.

\paragraph{Strong CAvoider-CAvoider games.}
In the last few years several variants of positional games have emerged, like the \emph{PrimMaker-Breaker game} introduced in~\cite{Prim} where the subgraph induced by Maker’s edges must be connected throughout the game. In the \emph{Walker-Breaker games} introduced by Espig, Frieze, Krivelevich, and Pegden~\cite{KrivelevichWalkerBreaker}, Maker is constrained to choose edges of a walk or a path. Similarly, in the \emph{WalkerMaker–WalkerBreaker games}, see~\cite{MimaJovanaWMWB}, both players have the constraint to claim edges of a walk.

In the second part of this paper we study \emph{Strong CAvoider-CAvoider games} in which the graph of each player must stay connected throughout the game. The board is still the edge set of $K_n$, and the players should not claim a copy of the forbidden graph. This is a natural extension of the strong Avoider-Avoider games, with a connectedness constraint analogue to the ones mentioned above.

\medskip

Let $S_3$ be a star on three leaves. In the following we prove that Blue can win in three different strong CAvoider-CAvoider games.

\begin{thm}\label{TH004}
Blue has a winning strategy in the Strong CAvoider-CAvoider $S_3$ game, played on $K_n$, where $n \geq 7$.
\end{thm}

\begin{thm}\label{TH003}
Blue has a winning strategy in the Strong CAvoider-CAvoider $P_4$ game, played on $K_n$, where $n \geq 5$.
\end{thm}
\vspace{0.2cm}

In the Cycle game the player who first claims a cycle loses.

\begin{thm}\label{TH005}
Blue has a winning strategy in the Strong CAvoider-CAvoider Cycle game, played on $K_n$, where $n \geq 6$.
\end{thm}
\vspace{0.2cm}

Note that if $F \in \{S_3, K_3\}$, $R(F)$ is equal to $6$, so draw is not possible in any of the three games.

\vspace{0.2cm}

The rest of the paper is organized as follows. In Section~\ref{s:prelim} we give notation and preliminaries. In Section~\ref{s:p4} we prove Theorem \ref{TH001}. Then, in Section~\ref{s:cc3} we prove Theorem \ref{TH002}. Finally, in Section~\ref{s:caca} we prove Theorem~\ref{TH004}, Theorem~\ref{TH003} and Theorem~\ref{TH005}.

\section{Preliminaries} \label{s:prelim}

During a game, we say that the vertices that are touched by Red are \emph{red vertices}, the ones touched by Blue are \emph{blue vertices} and the others, that are not touched by any of the players, are \emph{black vertices}. If a vertex is touched just by Red and not by Blue, we call it a \emph{pure red vertex}, and if the situation is opposite it is a \emph{pure blue vertex}.
\vspace{0.1cm}

\setlength{\parindent}{0pt}
 By a player's graph we consider the graph with all edges he claimed on the vertex set $V=[n]$.

 A \emph{star} is the complete bipartite graph $K_{1,k}$, where $k\geq 0$.
 We will refer to the star centered in $v$ as a \emph{$v$-star}. When we say that a player \emph{star-adds} a vertex $x$ to a $v$-star, this means that he claims the edge $vx$.
 A $P_n$ is a path on $n$ vertices.
 \vspace{0.1cm}

 We will use the abbreviation $RC$ for non-trivial red components, i.e. connected components in Red's graph, where we do not count isolated vertices as $RC$.
We will say that a connected component is \emph{pure red} (respectively, \emph{pure blue}) if all its vertices are pure red (respectively, pure blue).

\vspace{0.2cm}

We will make use of the following facts about the $P_4$-free graphs.

\begin{obs}\label{O1}
For every graph that does not contain a $P_4$ as a subgraph, its connected components can be stars and triangles (where we count isolated edges and vertices as stars).
\end{obs}


\begin{obs}
A $P_4$-free graph on $n$ vertices with the maximum number of edges is a disjoint union of triangles, when $n=3k$, for some integer $k$, and otherwise a disjoint union of one star and a number (possibly zero) of triangles. The number of edges in that graph is $n$, if $n=3k$, and $n-1$ otherwise.
\end{obs}

\begin{obs}\label{l3}
If in a maximal $P_4$-free graph there are $k$ stars, then it has $n-k$ edges.
\end{obs}

\vspace{0.2cm}

We also need the following facts about graphs that do not have connected components on more than three vertices.

\begin{obs}\label{o_1}
For every graph that does not contain a ${\cal CC}_{>3}$ as a subgraph, its connected component can be a triangle, a path on three vertices, an isolated edge or an isolated vertex.
\end{obs}

\begin{obs}\label{o1}
A ${\cal CC}_{>3}$-free graph with the maximum number of edges is a disjoint union of triangles, when $n=3k$, a disjoint union of triangles and one isolated vertex, when $n=3k+1$, or a disjoint union of triangles and one isolated edge, when $n=3k+2$, for some integer $k$.

The number of edges in that graph is $n$, if $n=3k$, and $n-1$ otherwise.
\end{obs}

\section{Strong Avoider-Avoider $P_4$ game} \label{s:p4}

\vspace{0.2cm}

{\bf Proof of Theorem \ref{TH001}:}
We will describe a winning strategy for Blue.
Note that by definition of a $RC$ and by Observation \ref{O1}, Red is not allowed to claim any edge between two $RC$ at any point of the game, as otherwise he would create a $P_4$ in his graph.

\medskip

In the beginning, we have a graph $G$ with $n$ isolated vertices, and Red claims an edge, let us denote it by $rt$. Then Blue claims an edge that is not adjacent to the red one, we denote it by $uv$.
In the following move Red has four options, up to isomorphism, for choosing an edge, and those four moves will make our four cases. For each of these cases we will show that Blue can win.
Let us denote the second move of Red by $e=xy$.

In the first three cases we use the idea of strategy stealing: we will suppose that at this point of the game (after Red played two moves and Blue played one) Red has a strategy to finish the game and win. Then we will show how Blue can use this strategy to win the game. That will lead to a contradiction, implying that our assumption was wrong and Blue can win the game.

\medskip
{\bf Case 1.} Vertex $x$ is red and $y$ is black.
\smallskip

Suppose that Red has a strategy $S$ to win the game. W.l.o.g let $x=t$. After Red plays $ty$ it is Blue's turn. The graph of the game consists of two adjacent red edges and one isolated blue edge. We denote the vertices as depicted in Figure \ref{f1sfig1}. Before his next move, Blue imagines that he has already claimed the edge $yu$ and that Red has not claimed the edge $ty$, see Figure \ref{f1sfig2}. Note that the edge $yu$ will remain free throughout the game, as otherwise Red would create a $P_4$ in his graph.

 \begin{figure}[ht]
    \begin{subfigure}{.5\textwidth}
    \centering
     \includegraphics[width=0.6\linewidth]{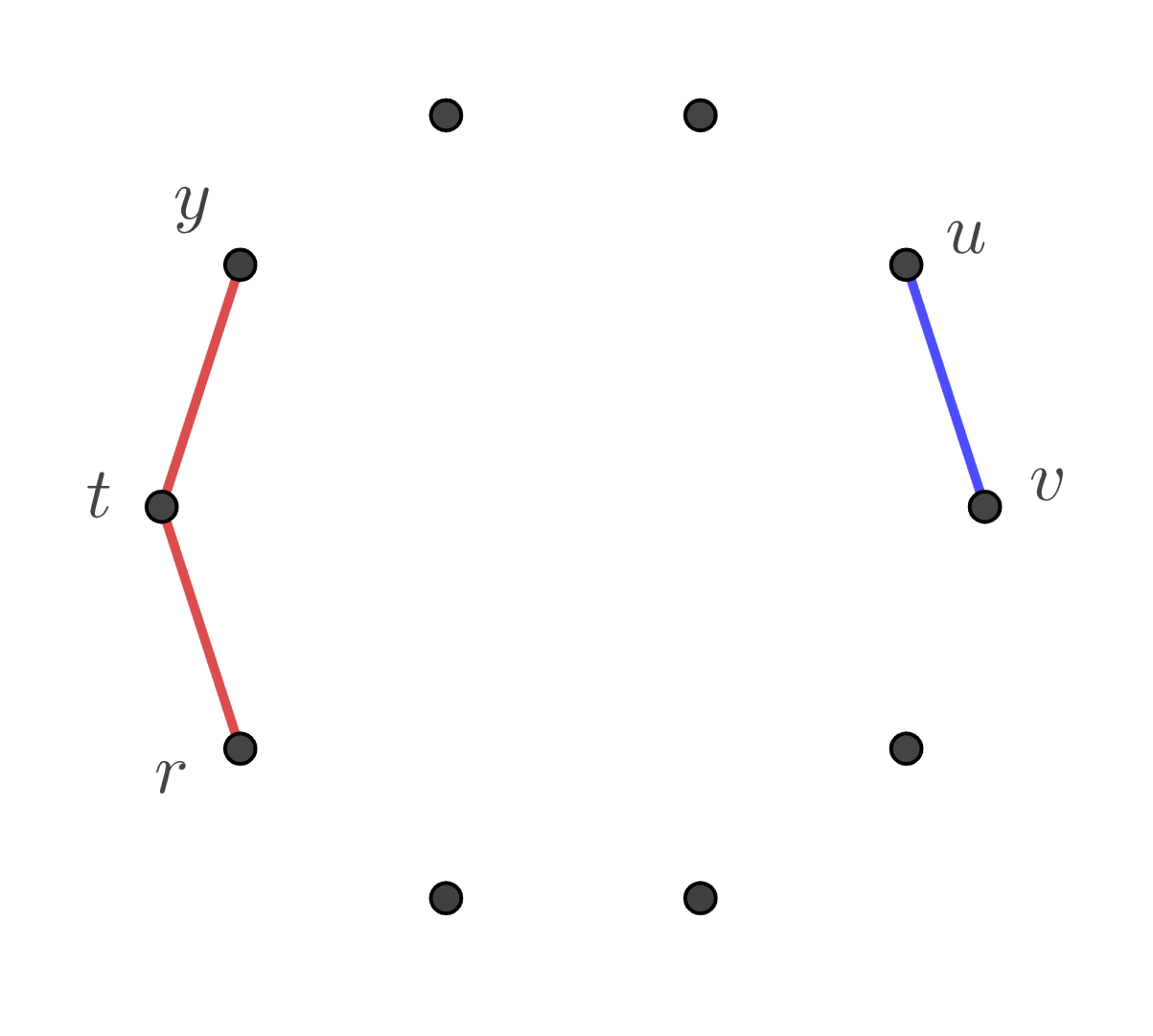}
     \caption{}
     \label{f1sfig1}
    \end{subfigure}%
    \begin{subfigure}{.5\textwidth}
     \centering
     \includegraphics[width=0.6\linewidth]{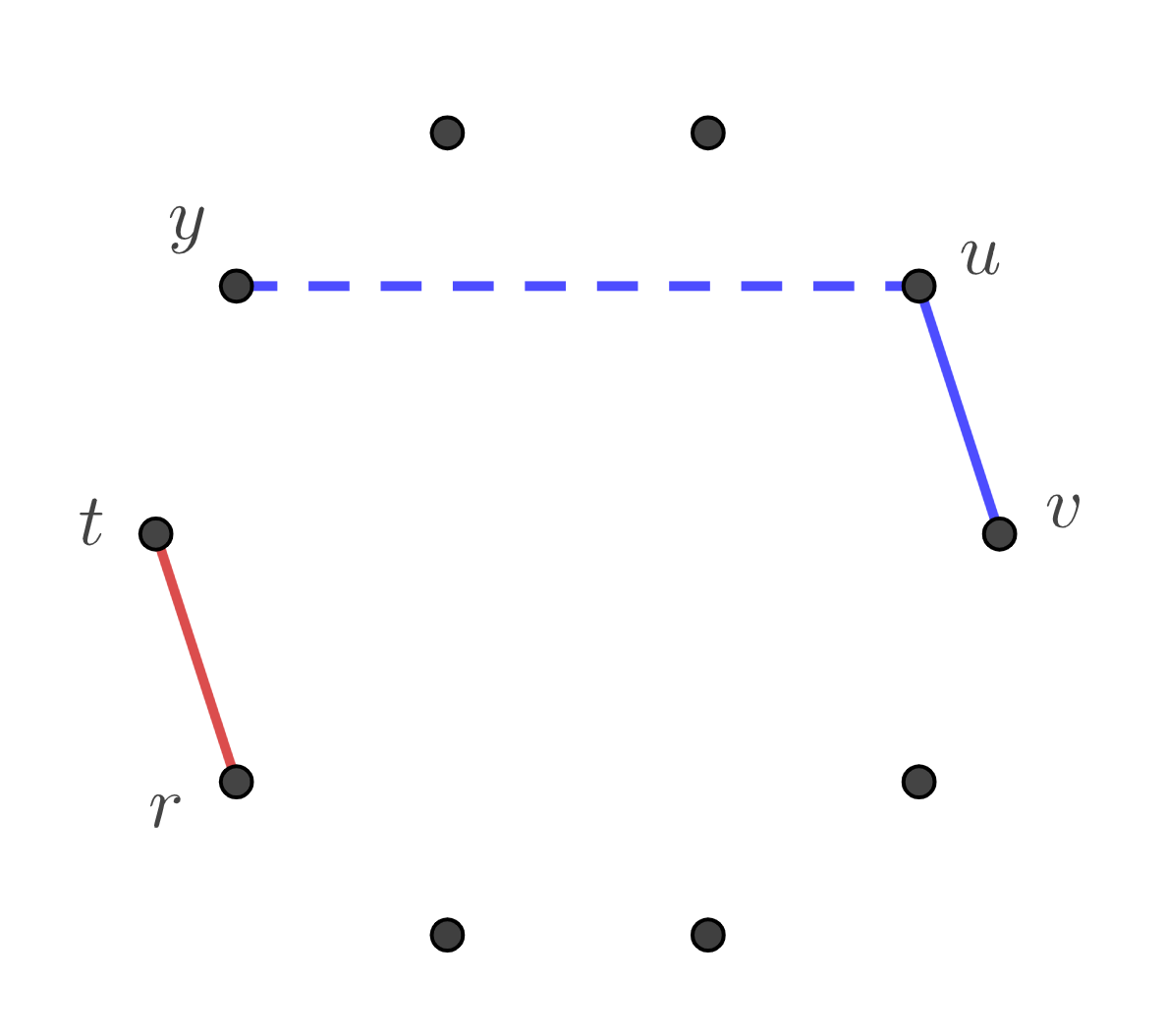}
    \caption{}
    \label{f1sfig2}
    \end{subfigure}
    \captionsetup{justification=centering}
    \caption{Case 1: (a) the graph before the second move of Blue. (b) The imagined graph before the second move of Red.}
    \label{f1}
    \end{figure}

 The imagined graph is isomorphic to the graph, where the roles of the players are swapped. Blue imagines that he is the first player, he further imagines that Red claims the edge $ty$ as his second move, and from now on responds as advised by the winning strategy $S$. Because this is a winning strategy, Blue wins the game, a contradiction.

\medskip
\noindent{\bf Case 2.} Vertex $x$ is red and $y$ is blue.
\smallskip

 Similar to Case 1, we suppose that Red has a strategy $S$ to win the game. W.l.o.g let $x=t$ and $y=u$. After Red plays $tu$ it is Blue's turn. The graph of the game consists of one $P_4$ with two adjacent red edges and one blue edge, see Figure \ref{f2sfig1}. Before his next move, Blue imagines that he has already claimed the edge $rv$ and that Red has not claimed the edge $tu$, see Figure \ref{f2sfig2}. Note that the edge $rv$ will remain free throughout the game, as otherwise Red would create a $P_4$ in his graph.

  \begin{figure}[ht]
    \begin{subfigure}{.5\textwidth}
    \centering
     \includegraphics[width=0.6\linewidth]{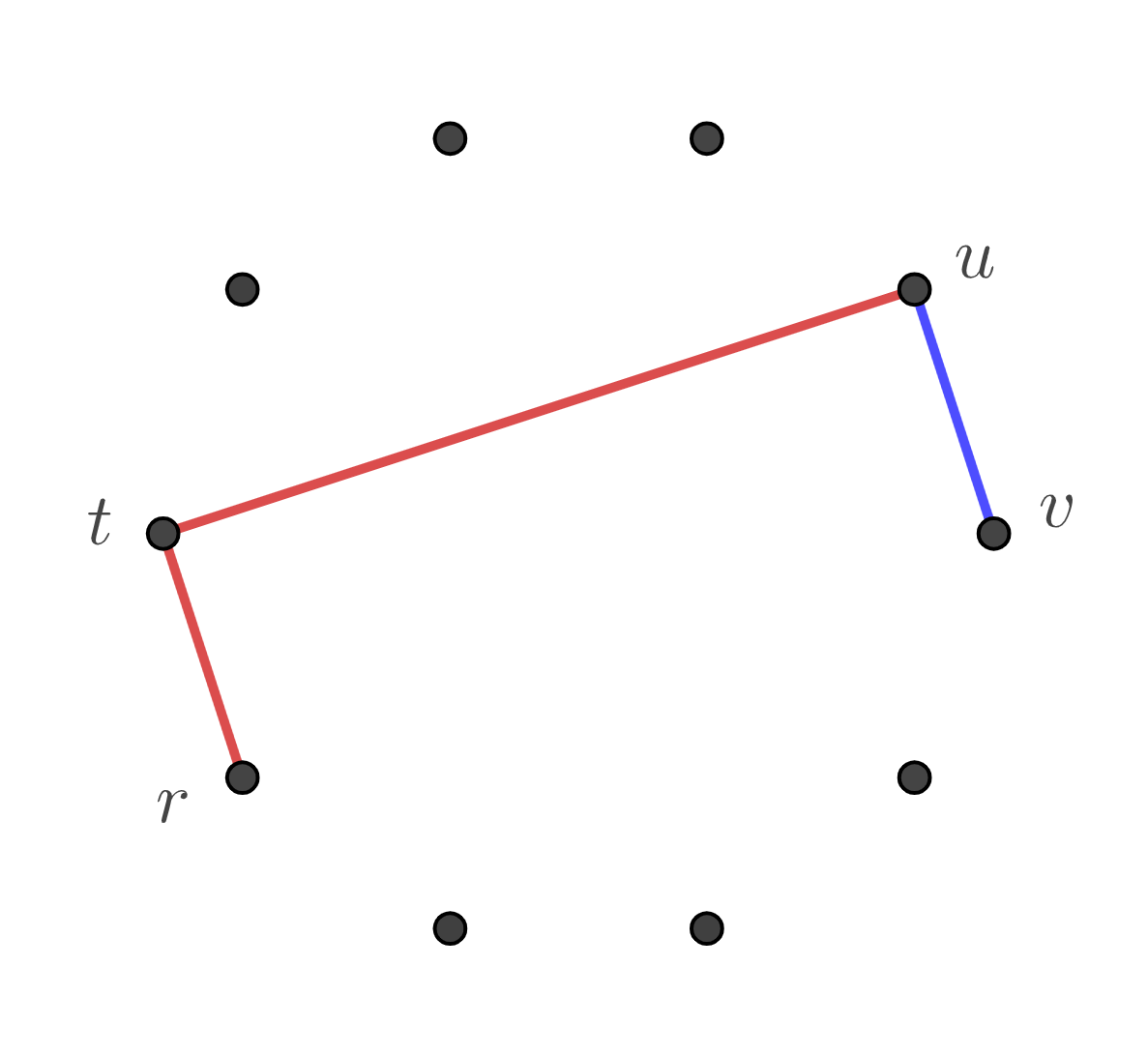}
     \caption{}
     \label{f2sfig1}
    \end{subfigure}%
    \begin{subfigure}{.5\textwidth}
     \centering
     \includegraphics[width=0.6\linewidth]{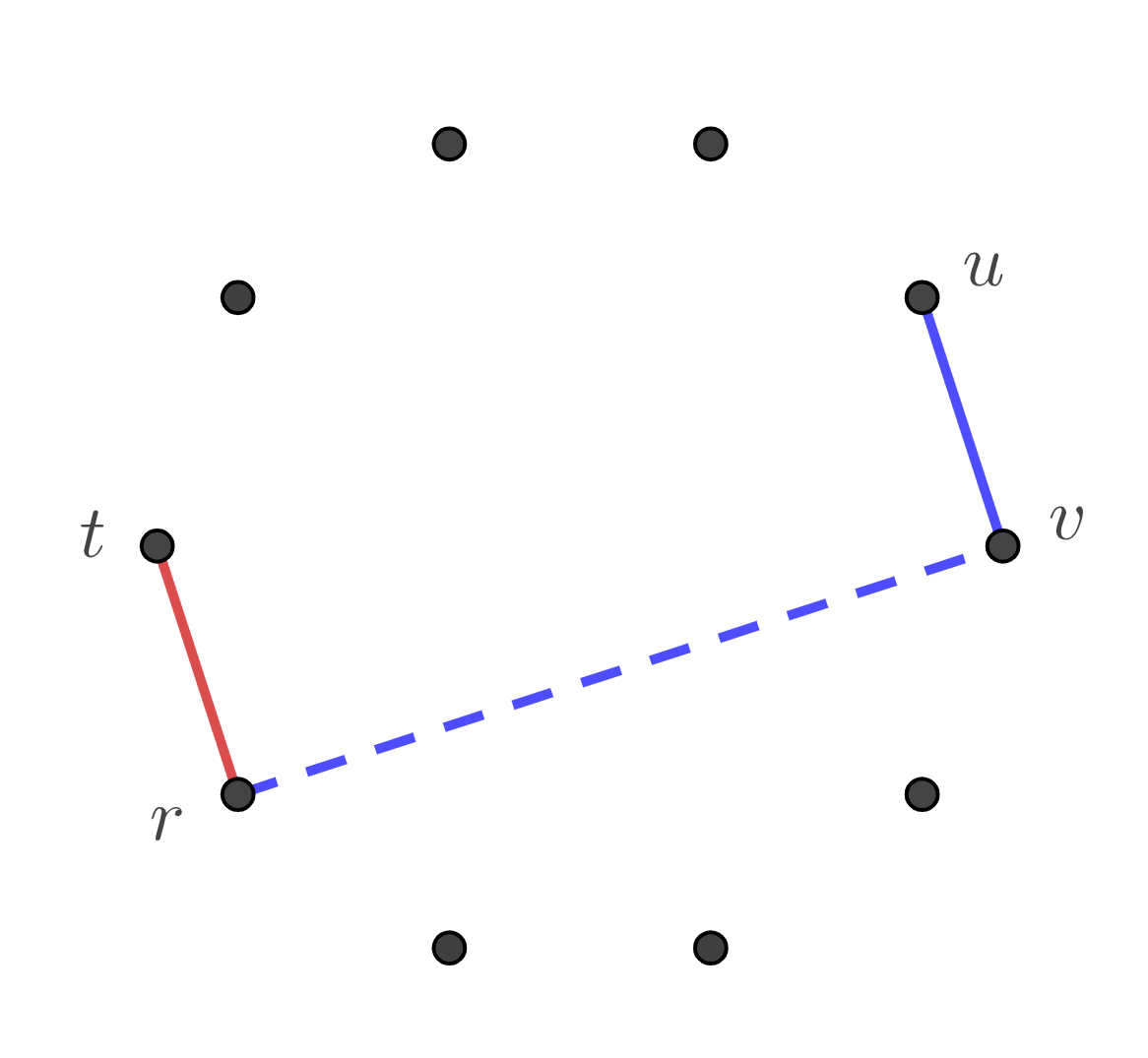}
    \caption{}
    \label{f2sfig2}
    \end{subfigure}
    \captionsetup{justification=centering}
    \caption{Case 2: (a) the graph before the second move of Blue. (b) The imagined graph before the second move of Red.}
    \label{f2}
    \end{figure}

 The imagined graph is isomorphic to the graph, where the roles of the players swapped. Blue imagines that he is the first player and that Red claims the edge $tu$ as his second move, and from now on Blue responds as advised by $S$ winning the game, a contradiction.

\medskip
\noindent{\bf Case 3.} Vertex $x$ is blue and $y$ is black.
\smallskip

 We again suppose that Red has a strategy $S$ to win the game. W.l.o.g let $x=u$.
 After Red plays $uy$ the graph of the game consists of one isolated red edge and one $P_3$ with two edges of different colours, see Figure \ref{f3sfig1}. Before his next move, Blue imagines that he has already claimed the edge $ty$ and that Red has not claimed $uy$, see Figure \ref{f3sfig2}. Note that the edge $ty$ will remain free throughout the game.

  \begin{figure}[ht]
    \begin{subfigure}{.5\textwidth}
    \centering
     \includegraphics[width=0.6\linewidth]{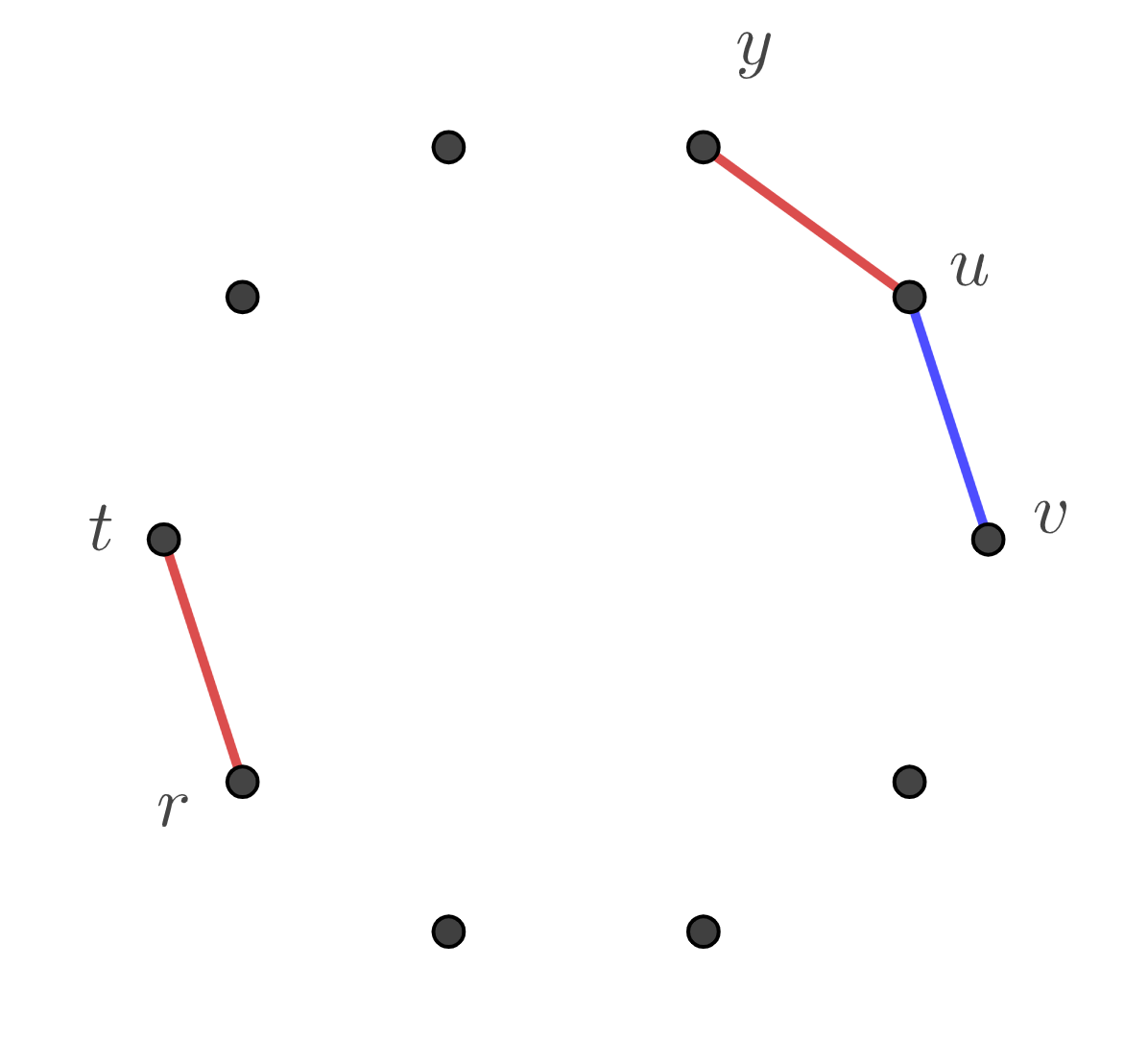}
     \caption{}
     \label{f3sfig1}
    \end{subfigure}%
    \begin{subfigure}{.5\textwidth}
     \centering
     \includegraphics[width=0.6\linewidth]{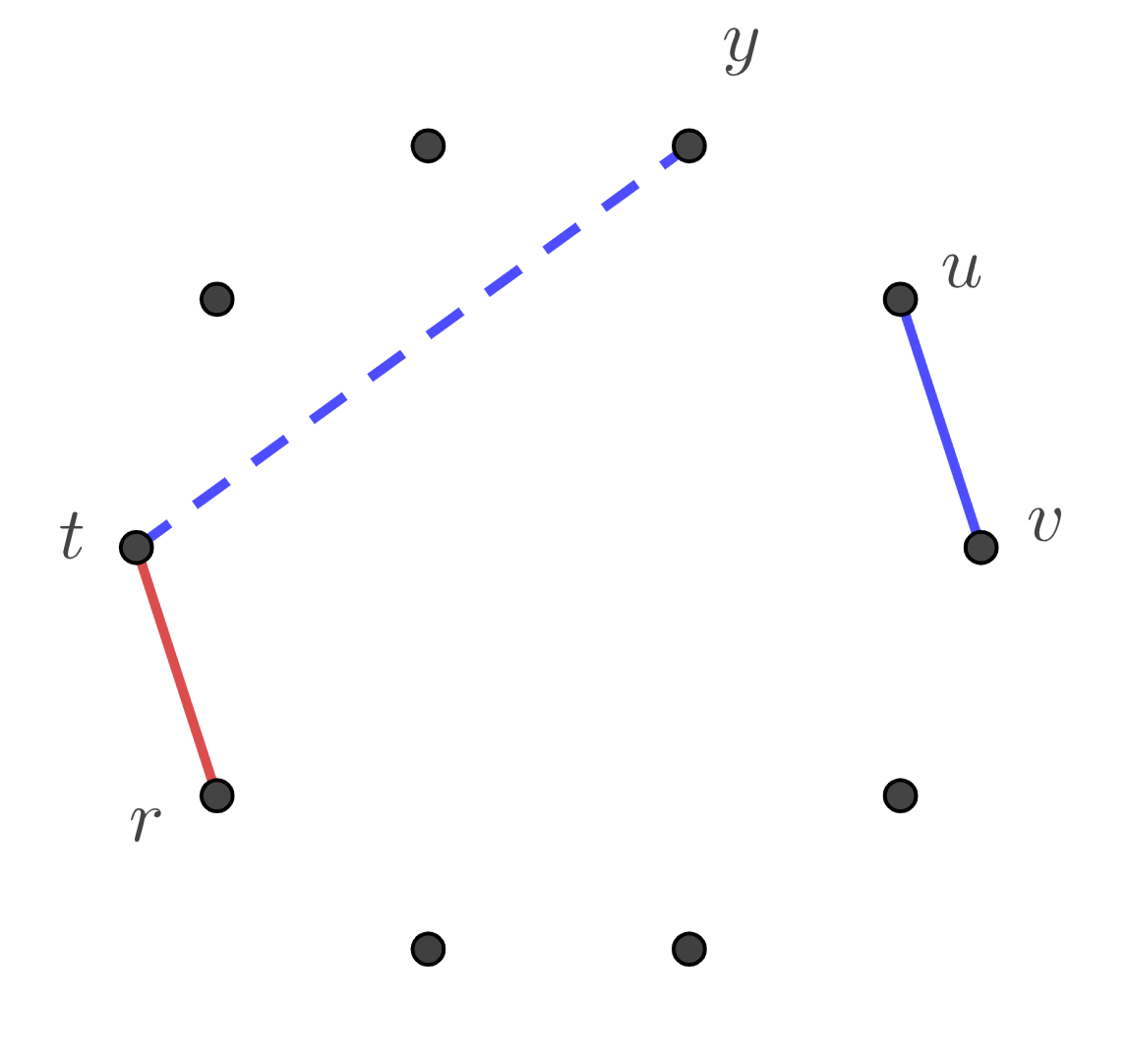}
    \caption{}
    \label{f3sfig2}
    \end{subfigure}
    \captionsetup{justification=centering}
    \caption{Case 3: (a) the graph before the second move of Blue. (b) The imagined graph before the second move of Red.}
    \label{f3}
    \end{figure}

 The imagined graph is isomorphic to the graph, where the roles of the players are swapped. Blue imagines that he is the first player and that Red claims the edge $uy$ as his second move. From now on, Blue responds as advised by $S$ thus winning the game, a contradiction.

\medskip
  {\bf Case 4.} Both vertices $x$ and $y$ are black.
\smallskip

 The Red's graph at this moment has two isolated edges that make the first two $RC$. Let us denote by $C_1$ the component $\{r,t\}$, and by $C_2$ the component $\{x,y\}$. For the reminder of the game, we will dynamically update $C_1$ and $C_2$ as they grow. Note that $C_1$ will remain a different $RC$ from $C_2$.

 Blue is the second player, so his graph cannot have more edges than the Red's graph. Having that in mind, as well as Observation \ref{l3}, we will describe a strategy for Blue to keep the number of stars in his graph less then or equal to the same number in the Red's graph throughout the game.

 After Red claims the edge $xy$, Blue responds by claiming the edge $vr$, as depicted in Figure \ref{f4sfig1}. Note that at this moment the Blue's graph consists of one $v$-star. In the rest of the game Blue will enlarge this $v$-star, and possibly create isolated triangles.

   \begin{figure}[ht]
    \begin{subfigure}{.5\textwidth}
    \centering
     \includegraphics[width=0.65\linewidth]{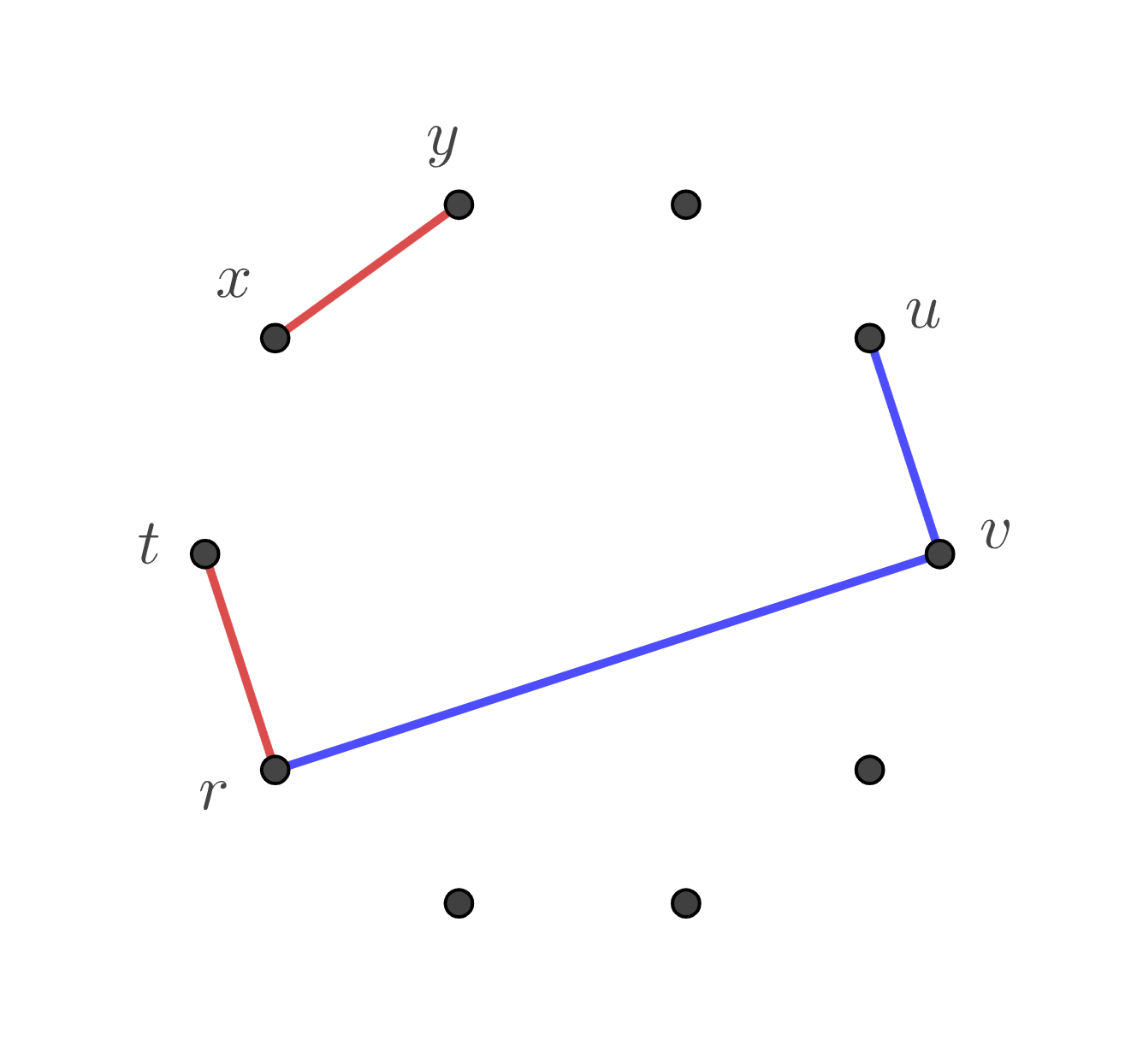}
     \caption{}
     \label{f4sfig1}
    \end{subfigure}%
    \begin{subfigure}{.5\textwidth}
     \centering
     \includegraphics[width=0.55\linewidth]{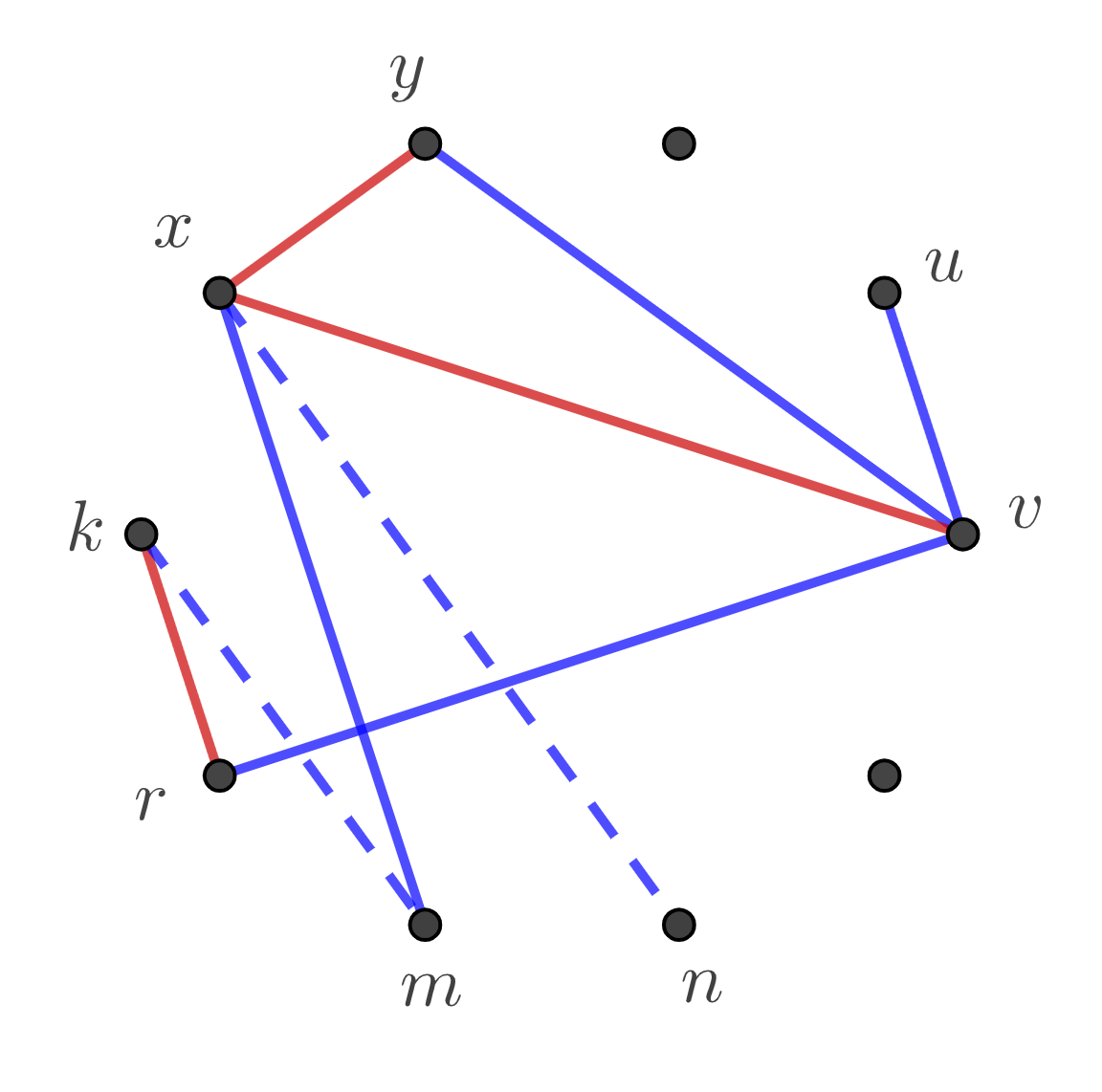}
    \caption{}
    \label{f4sfig2}
    \end{subfigure}
    \captionsetup{justification=centering}
    \caption{Case 4: (a) the graph after the second move of Blue. (b) The possible moves of Blue if the rule 1 of Stage 1 is in order, shown as dashed lines.}
    \label{f4}
    \end{figure}




During the game, every vertex $k$ that is not blue, for which it applies that $kv$ is free and adding the edge $kv$ to the Red's graph will not make a $P_4$ will be called a \emph{dangerous vertex}. All the other vertices will be called \emph{safe}.
A pure red vertex $j$ that is adjacent to the vertex $v$ in Red's graph will be called \emph{inaccessible}.

If in his third move Red claims the edge $yv$, Blue responds by claiming the edge $xv$, otherwise he claims the edge $yv$.
W.l.o.g.~we will suppose that Blue has claimed the edge $yv$ in his third move.

Note that at this point there are only two vertices in $C_1 \cup C_2$, namely $x$ and $t$, that can be dangerous.
Let $S_1 := $ \{$x, t\} $. During the game, whenever a vertex from $S_1$ becomes blue we remove it from $S_1$.

\medskip

 Now we give a strategy for Blue that he follows from his fourth move on.

 {\bf Stage 1.} While there are at least two black vertices in the game, Blue repeatedly plays by the first rule in this list that is applicable.
 \begin{enumerate}
     \item\label{SI1} If Red has claimed one of the edges $vx$ or $vt$;\\
     in his following three moves, Blue will claim edges that close a triangle incident with the vertex that just become inaccessible, w.l.o.g.~let it be the vertex $x$.
     By $m$ and $n$ we denote two arbitrary black vertices and by $k$ the other vertex from $S_1$. Note that the only case when $k$ is not vertex $t$ is when rule 4 has been played before.
     Blue starts by claiming the edge $xm$. Then, if it is unclaimed he claims the edge $km$, otherwise the edge $xn$, and in the following move Blue closes either the triangle $xmk$ or $xmn$, see Figure \ref{f4sfig2}.
     Then, he star-adds all the remaining vertices of the base graph that are not blue to the $v$-star.

     \item\label{SI2} If Red claims an edge creating a $RC$ that is a star on three vertices with the vertex $v$ as a leaf, and if the edge incident with $v$ and the other leaf is unclaimed;\\
     Blue claims it. We denote by $r$ the inaccessible vertex, the center of the red star. In his following three moves Blue claims the edges of the triangle $rxt$. Then, he star-adds all the remaining vertices of the base graph that are not blue to the $v$-star.

     \item\label{SI3} If Red claims an edge creating a $RC$ that is a $v$-star on three vertices, and if the edge incident with both leaves is unclaimed;\\
     Blue claims that edge. That isolated blue edge we call a \emph{cover-edge}.

     \item\label{SI4} If there is exactly two black vertices, and there is no pure red vertex that is not in $C_1 \cup C_2$, and there is no cover-edge, and one of $\{C_1, C_2\}$ is an isolated edge while the other one is a star with at least three edges that does not have $v$ as a leaf;\\
     then Blue claims the edge incident with the center of that star and $v$. Now, we remove the blue vertex from $S_1$ and add a safe pure red vertex from the same $RC$ to $S_1$.

     \item\label{SI5} Otherwise;\\
     Blue claims an edge incident with $v$ and one black vertex.
 \end{enumerate}

   Now, we will prove that if it is Blue's turn to play Stage 1, he can follow it.
  First, note that if in his third move Red claimed the edge $yv$, the game would be finished in step 1 of Stage 1.

   If Red did not claim $yv$ in his third move, in the beginning of Stage 1 Blue's graph consists of the $v$-star and isolated vertices, and he will continue claiming the edges of the $v$-star using rules \ref{SI4} and \ref{SI5}, until a condition of one of the rules \ref{SI1}, \ref{SI2}, \ref{SI3} is fulfilled.
   Note that Blue can use exactly one of the rules \ref{SI1}, \ref{SI2}, \ref{SI3} at most once until the end of the game, so when it is Blue's turn to play one of them, his graph consists of the $v$-star and isolated vertices.
   Clearly, Blue can always claim an edge between $v$ and a black vertex.

   When Blue is to play by rule \ref{SI1}, we know that none of rules $\ref{SI1},\ref{SI2}$ and $\ref{SI3}$ have been used before. Therefore Blue's graph consists of the $v$-star and isolated vertices. Also, there are two black vertices and Red has made one inaccessible vertex $x$. The edge $vx$ will be the only red edge incident with $v$ until the end of game, because vertex $v$ is a leaf of a red star and claiming another edge incident with $v$ would make a $P_4$ in Red's graph. So, all the remaining vertices that are not blue are safe. It is clear that $k$ is pure red and not in the same RC as $x$. Hence, Blue can follow rule \ref{SI1} and play up to $n-1$ moves, thus winning.

   When it is Blue's turn to play by rule \ref{SI2}, first unclaimed edge advised by the strategy must be available for him because his graph has the $v$-star and every vertex not adjacent to $v$ is isolated in Blue's graph. In his following three moves Blue can claim the edges between $r$, $x$ and $t$, because each of them is in a different $RC$. Note that $x$ and $t$ are pure red because rule \ref{SI4} could not have happened before and Blue could use only rule \ref{SI5}. For the same reason as above, all the remaining vertices that are not blue are safe. Now, it is clear that Blue can follow rule \ref{SI2} and play up to $n-1$ moves, thus winning.

   When it is Blue's turn to play by rule \ref{SI3} none of rules $\ref{SI1},\ref{SI2}$ and $\ref{SI3}$ have happened before, so every vertex not adjacent to $v$ is isolated in Blue's graph and he can claim the cover-edge as advised by the strategy.
   Note that Red cannot ever claim any edge adjacent to the cover-edge.

   For further analysis we need to verify the following claim.

    \begin{cl}\label{cs12}
      From the moment in the game when there is no more than two black vertices, until the first Blue's move after Stage 1, if there is no pure red vertex that is not in $C_1 \cup C_2$ and there is no cover-edge, Red's graph has at least four edges and one safe pure red vertex in the $C_1 \cup C_2$, and at least one pure red vertex in each of these components.
   \end{cl}
   \begin{proof}
   Before Blue had played his fourth move, his graph consisted of the $v$-star on 4 vertices, where $v$ was pure blue. At that moment Red's graph had four edges and all of them had to be in $C_1 \cup C_2$, otherwise there would be at least one pure red vertex in the third $RC$, which we assumed was not the case. Therefore there are three options for $C_1$ and $C_2$:
   \begin{itemize}
       \item An isolated edge and a triangle.\\
       In this case there were at least two pure red vertices, one in $C_1$ and the other one in $C_2$, where one of them was incident with a triangle, so it must have been safe.
       \item An isolated edge and a star on four vertices.\\
        In this case there were at least three pure red vertices, one of which was incident with the isolated edge, and all the others with the star. Therefore, there were at least one safe pure red vertex as a leaf of the red star.
       \item Both of them are a $P_3$.\\
        In this case there were at least three pure red vertices, at most two of them were dangerous, so there must have been one safe. At least one vertex in each component was pure red.
   \end{itemize}

   If the assumption of the claim holds, the only rules that Blue could have applied in the meantime are rules \ref{SI4} and \ref{SI5}. The last one does not have any influence on pure red vertices, and rule \ref{SI4} can just swap one pure red dangerous vertex with a pure red safe vertex in the same component. Therefore, the assertion of the claim is proven.
   \end{proof}

   When Blue is to play by rule \ref{SI4}, it is clear that he can claim that edge. Note that if that edge is not free it has to be blue, otherwise rule \ref{SI1} would be achieved. Using Claim \ref{cs12} we know that a pure red vertex incident with the star exists.

   Note that during Stage 1, if Blue has not already won (rule \ref{SI1} and \ref{SI2}), his graph consists of the $v$-star (rule \ref{SI4} and \ref{SI5}), possibly one isolated cover-edge (rule \ref{SI3}) and isolated vertices. Also, $S_1$ consists of two pure red vertices where one belongs to $C_1$, and the other one to $C_2$.

   When Stage~1 is finished, there is at most one black vertex. We then move on to Stage 2, distinguishing two cases.

   \medskip
   {\bf Stage 2a.} If there is at least one pure red vertex that is not in $C_1 \cup C_2$ or there is a cover-edge, we proceed to Stage 2a.
   \smallskip

   Before Blue plays his first move in Stage 2a, we add all inaccessible vertices and the ends of the cover-edge to $S_1$.
    If a pure red vertex that is not in $C_1 \cup C_2$ exists, we denote it by $w$.
    If $|S_1|<3$ (there was not an inaccessible vertex nor a cover-edge) then we add $w$ to $S_1$.

    Blue repeatedly plays by the first rule that is applicable in this list and if before the move of Blue there is a new inaccessible vertex, we add it to $S_1$.
    \begin{enumerate}
        \item\label{SII1a}  If the conditions of rule \ref{SI2} or rule \ref{SI3} from Stage 1 are fulfilled,\\
        Blue claims the next edge in the same way as that rule suggests.
        \item\label{SII1b} If there is a black vertex,\\
        Blue star-adds it to the $v$-star.
        \item\label{SII1c} If there is a cover-edge,\\
        then if $|S_1|$ is not divisible by three, we will make it by removing one or two vertices from $C_1 \cup C_2$. In his following two moves Blue claims a triangle using the cover-edge and one more vertex from $S_1$. Then, until $S_1$ is not empty, he chooses three vertices from $S_1$ and makes a triangle claiming all edges between them. At the end he star-adds all the remaining vertices of the graph that are not blue to the $v$-star.
        \item\label{SII1d} If there is an inaccessible vertex $r$,\\
        then if $|S_1|=4$ we remove $w$ from $S_1$. In his following three moves Blue makes the triangle claiming edges between vertices from $S_1$. Then he star-adds all the remaining vertices that are not blue to the $v$-star.

        \item\label{SII1e} If there is an unclaimed edge incident with $v$ and one pure red dangerous vertex that is not in $S_1$,\\
        Blue claims it.

        \item\label{SII1f} Otherwise,\\
        in his following three moves Blue makes the triangle claiming edges between the remaining three vertices from $S_1$. Then he star-adds all the remaining vertices that are not blue to the $v$-star.
    \end{enumerate}

   \medskip
   {\bf Stage 2b.} Otherwise (there is neither a pure red vertex that is not in $C_1 \cup C_2$ nor a cover-edge), we proceed to Stage 2b.
   \smallskip

   If there is one black vertex, we denote it by $j$.
   Depending on the types of the components $C_1$ and $C_2$, we have three conditions and Blue chooses the first one which is satisfied.
   \begin{enumerate}
       \item\label{SII2a} At least one of $C_1$ and $C_2$ is a star with more than two edges, and it is disjoint from $v$.\\
       Let us denote by $r$ the center of that star. If it is unclaimed, Blue claims the edge $vr$. Than, he star-adds all the remaining vertices that he can to the $v$-star.

       \item\label{SII2b} Each of $C_1$ and $C_2$ is a star with at least two edges.
       \begin{enumerate}
           \item If $v$ is red,\\
           we denote by $r$ the center of the star incident with $v$, and by $k$ a pure red vertex that is not in the same $RC$ as $r$. We know that these vertices exists by Claim \ref{cs12}. Blue claims the edge $rj$, after that if it is free he claims the edge $kj$, and then the edge $kr$. Then he star-adds all the remaining vertices that are not blue to the $v$-star.
           \item Otherwise, $v$ is blue,\\
           we denote by $w$ a safe pure red vertex and by $r$ the center of the star of the same $RC$, and with $k$ a pure red vertex from the other $RC$. We know that these vertices exists by Claim \ref{cs12}.\\
           If it is unclaimed, Blue claims the edge $rv$.
           Then, if $kj$ is unclaimed Blue claims it. In his following three moves, he claims the triangle $kjw$ and star-adds all the remaining vertices that are not blue to the $v$-star, if any. Otherwise, if $kj$ is not unclaimed, he claims the edge $kv$ and star-adds all the remaining vertices that are not blue to the $v$-star.
       \end{enumerate}

       \item\label{SII2c} At least one of $C_1$ and $C_2$ is a triangle.\\
       We denote by $k$ a pure red vertex incident with the triangle, and with $r$ a pure red vertex from the other component, where if there are more than one such vertex the dangerous one has an advantage. We know that these vertices exists by Claim \ref{cs12}.

       Blue claims the edge $rj$, if it is unclaimed, and then creates the triangle $rjk$, otherwise he claims the edge $rv$. Then he star-adds all the remaining vertices that are not blue to the $v$-star.
        \end{enumerate}

    Now let us first show that when it is Blue's turn to play Stage 2a, he can follow it and win.
   Note that there are no red edges between any two vertices of $S_1$ because they are in two different $RC$ or they are leaves of the same red star, and all vertices in $S_1$ are pure red. Also, when it is Blue's turn to play rules \ref{SII1c}-\ref{SII1f} there are no more black vertices.

   When it is Blue's turn to play rule \ref{SII1a}, rules \ref{SII1a}, \ref{SII1c}, \ref{SII1d} and \ref{SII1f} could not have been activated before, so his graph consists of the $v$-star and isolated vertices. For the same reason as in rules \ref{SI2} and \ref{SI3} from Stage 1 he can claim his next edge.

   When it is Blue's turn to play rule \ref{SII1b}, he can obviously follow it.

   When it is Blue's turn to play by rule \ref{SII1c}, there are no more black vertices, so red star centered in $v$ cannot spread any more as all pure blue vertices are in the $v$-star, so all the remaining vertices that are not blue have to be safe. Now it is evident that Blue can follow his strategy as described in rule \ref{SII1c}. Here, Blue wins by playing $n-1$ edges.

   When it is Blue's turn to play by rule \ref{SII1d}, we know that $v$ is a leaf of a red star. At this moment Blue's graph consists of the $v$ star and isolated vertices, and all the remaining vertices that are not blue have to be safe. Now it is clear that Blue can follow his strategy and win by playing $n-1$ edges.

   It is obvious that if it is Blue's turn to play by rule \ref{SII1e}, he can claim as advised due to the definition of a dangerous vertex. Note that here $v$ is blue.

   If nothing above mentioned happened, $v$ is still blue and the Blue's graph consists of the $v$ star and isolated vertices. $S_1$ consists of three vertices, where each of them is in a different $RC$. Obviously, Blue can make the triangle described in \ref{SII1f}, and because these were the last dangerous vertices, he can star-add all of the remaining vertices to the $v$-star and win with $n-1$ edges.
   Note that it is not possible that Red claims a triangle incident with $v$ in this step because $v$ is blue, and the vertices from $S_1$ cannot be adjacent in Red's graph.

   Taking into account that Red's graph cannot have a triangle incident with $v$ (considering rule \ref{SII1a} and the above mentioned), Red's graph cannot have more than $n-1$ edges, so Blue wins the game.

\medskip
    It remains to show that when it is Blue's turn to play Stage 2b, he can follow it and win.
    Note that all pure red vertices are in $C_1 \cup C_2$, so there are at most three dangerous vertices and each of them has to belong to the set $\{x,t,j\}$. Likewise, all the vertices that are not in $C_1 \cup C_2$ are blue, except $j$ which is black (if it exists). Each of the blue vertices is a leaf of the $v$-star, therefore it is not possible that Red claims a triangle incident with $v$.

   When it is Blue's turn to play by rule \ref{SII2a}, it is clear that Blue can claim the edge $vr$ if it is free, and then Blue can star-add to the $v$-star all the remaining vertices but possibly one. In that case Red has at least two stars (one $r$-star and the other one incident with $v$) in his graph and he cannot have more than $n-2$ edges, by Observation \ref{l3}, so Blue wins with $n-2$ edges.\\
   Otherwise, if $vr$ is not free it has to be blue (condition of this step), so Blue just skips this move and wins in the same way as argued above.

    When it is Blue's turn to play by rule \ref{SII2b}:
    \begin{itemize}
     \item If $v$ is red, that happened in the last move, otherwise the game would have be finished in Stage 1, so there has to exist $j$, and Blue can claim $rj$. After that all the vertices that are not blue are safe.\\
     Then, if the edge $kj$ has been claimed Red will have at least two stars at the end of game, so he can have at most $n-2$ edges, by Observation \ref{l3}. Therefore, Blue can follow rule \ref{SII2b} to the end and win with $n-2$ edges.\\
     Else, if $kj$ is unclaimed, Blue claims the triangle $rjk$, and wins with $n-1$ edges.

     \item If $v$ is blue, each of $C_1$ and $C_2$ is a $P_3$, as otherwise it would be rule \ref{SII2a}. There has to exist $j$, otherwise if Red took it in his last move, before that move one of $C_1, C_2$ was an isolated edge, and the other one $P_3$, and that is not possible because of Claim \ref{cs12}.\\
    If the edge $rv$ has been already claimed, it has to be blue. It is clear that in his following move he can claim one of the edges $kj$ or $kv$, and then all the remaining pure red vertices are safe, so he can follow his strategy until the end of the game and win with $n-1$ edges.
   \end{itemize}

    When it is Blue's turn to play by rule \ref{SII2c}, we know that his previous move was in Stage 1, so after that move there was at least one black vertex and $v$ was blue. Red could not claim both of the edges $rj$ and $rv$ in his following move, so one of them is unclaimed and Blue can claim it. All the remaining pure red vertices are safe, so he can follow his strategy until the end of the game and win with $n-1$ edges.

    Note that Claim \ref{cs12} guarantee that all cases are covered by Stage 2b except the case when one component is an isolated edge and the other one a star with at least four edges that are incident with $v$. That cannot happen because in his previous move Blue played in Stage 1 and the conditions of rule~\ref{SI4} had to be fulfilled, but then Blue would make the center of that star adjacent to $v$.

    This concludes the proof for Case 4. \hfill $\Box$

\section{Strong Avoider-Avoider ${\cal CC}_{>3}$ game} \label{s:cc3}

\vspace*{3mm}

{\bf Proof of Theorem \ref{TH002}:}
 We describe a winning strategy for Blue. In the beginning, we have a graph $G$ with $n$ isolated vertices, and Red claims an edge, let us denote it by $uv$. Then Blue claims an edge that is adjacent to the red one, let us denote it by $vi$.

In the following move Red has five options, up to isomorphism, for choosing an edge, and those five moves will make our five cases. For each of these cases we will show that Blue can win, in the first four cases we will use the idea of Strategy stealing, and in Case 5 we will design an explicit strategy.
Let us denote the second move of Red by $e=xy$.

\medskip
{\bf Case 1.} Vertex $x$ is pure red and vertex $y$ is black, i.e. $x=u$.

Suppose that Red has a strategy $S$ to win the game. After Red plays $uy$ it is Blue's turn. The graph of the game has one $P_4$ with two adjacent red edges and one blue edge. We denote the vertices as depicted in Figure \ref{f5sfig1}. Before his next move, Blue imagines that he has already claimed the edge $ui$, and that Red has not claimed the edge $uv$, see Figure \ref{f5sfig2}. Note that the edge $ui$ will remain free throughout the game, as otherwise Red would create a ${\cal CC}_{>3}$ in his graph.

 \begin{figure}[ht]
    \begin{subfigure}{.5\textwidth}
    \centering
     \includegraphics[width=0.6\linewidth]{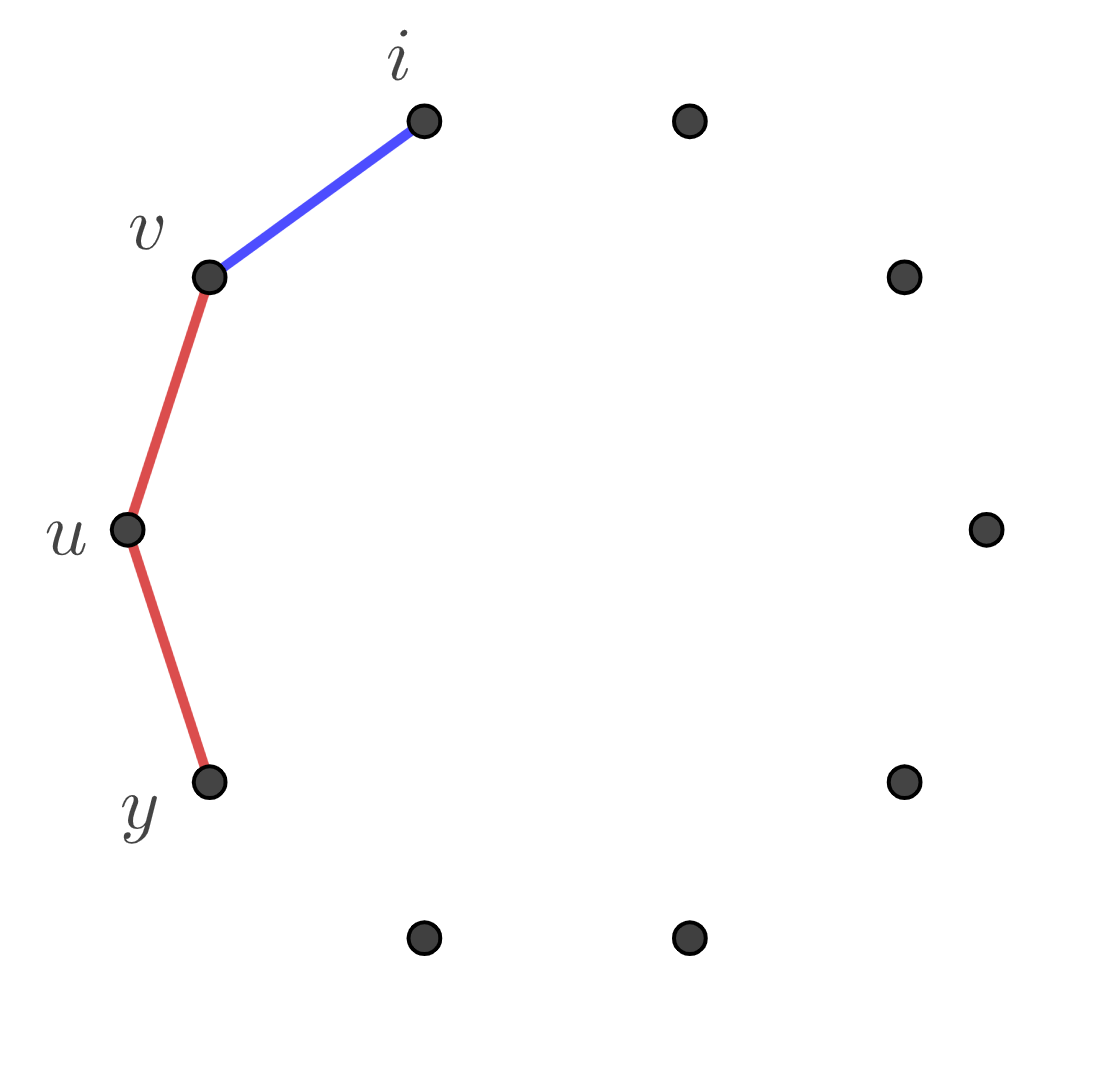}
     \caption{}
     \label{f5sfig1}
    \end{subfigure}%
    \begin{subfigure}{.5\textwidth}
     \centering
     \includegraphics[width=0.6\linewidth]{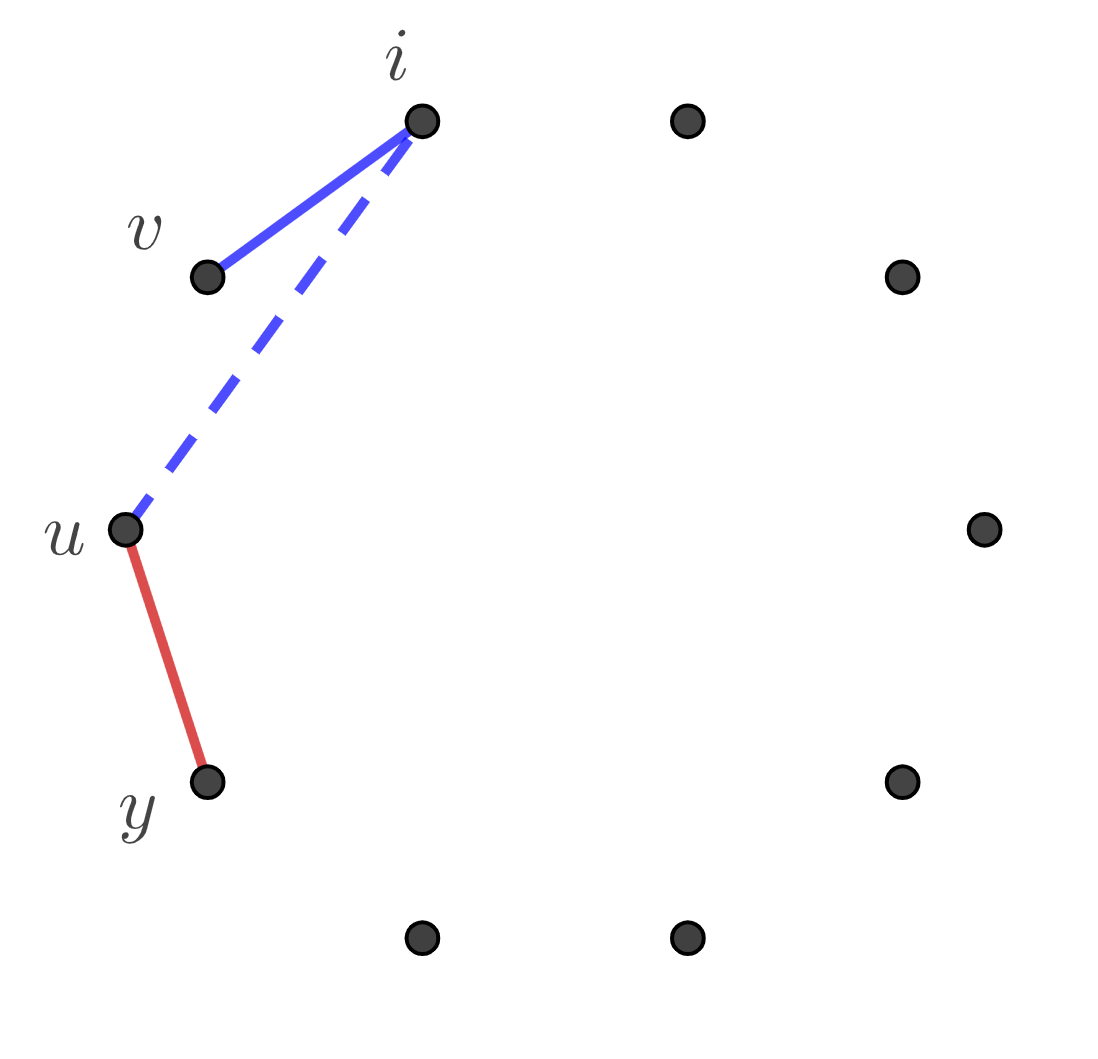}
    \caption{}
    \label{f5sfig2}
    \end{subfigure}
    \captionsetup{justification=centering}
    \caption{Case 1: (a) the graph before the second move of Blue. (b) The imagined graph before the second move of Red.}
    \label{f5}
    \end{figure}

 The imagined graph is isomorphic to the graph where the roles of the players are swapped. Blue imagines that he is the first player, he further imagines that Red claims the edge $uv$ as his second move, and from now on responds as advised by the winning strategy $S$ and wins the game, a contradiction.

\medskip
{\bf Case 2.} Vertex $x$ is both red and blue and vertex $y$ is black, i.e. $x=v$.

 \begin{figure}[ht]
    \begin{subfigure}{.5\textwidth}
    \centering
     \includegraphics[width=0.6\linewidth]{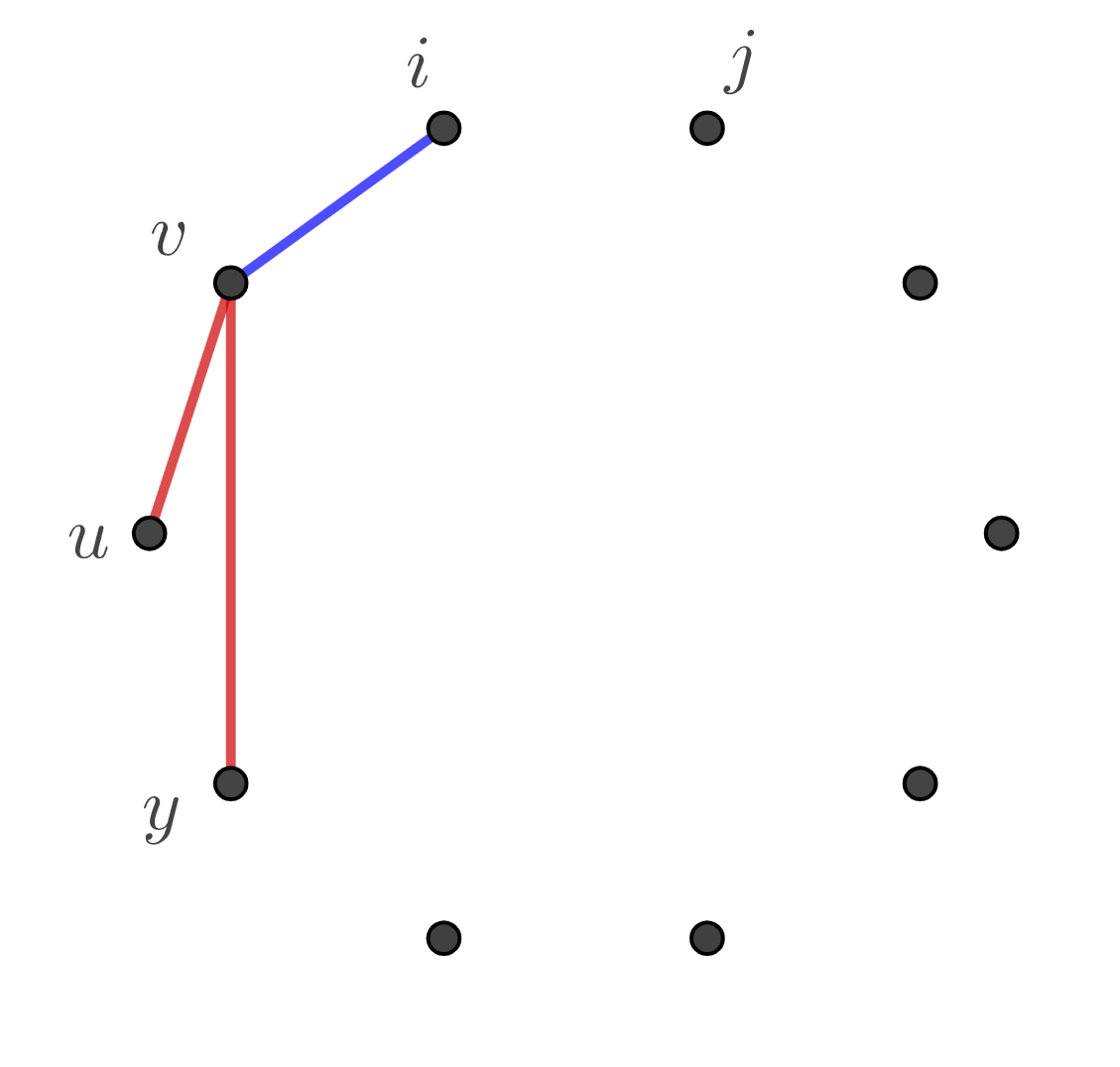}
     \caption{}
     \label{f6sfig1}
    \end{subfigure}%
    \begin{subfigure}{.5\textwidth}
     \centering
     \includegraphics[width=0.6\linewidth]{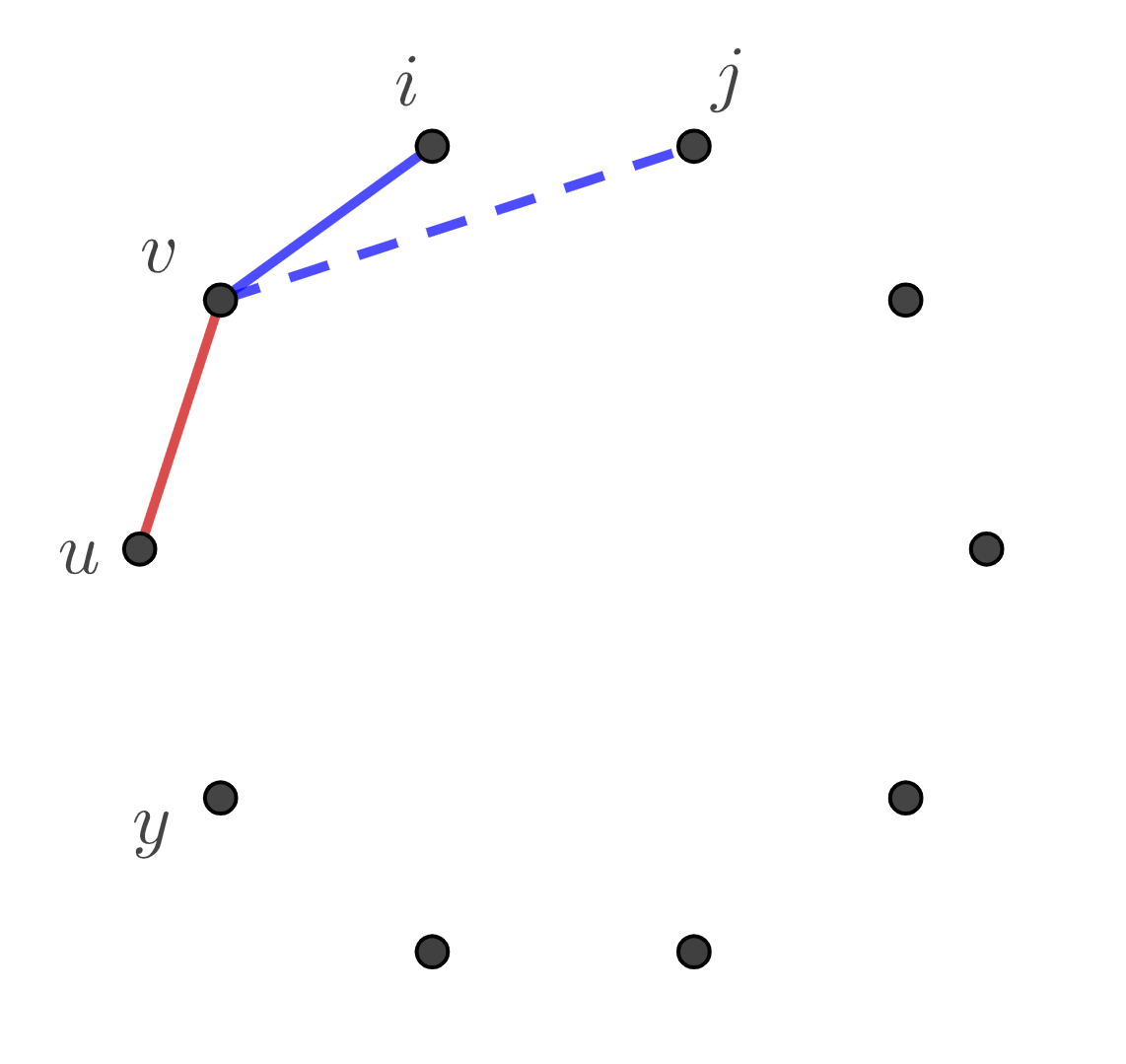}
    \caption{}
    \label{f6sfig2}
    \end{subfigure}
   \captionsetup{justification=centering}
   \caption{Case 2: (a) the graph before the second move of Blue. (b) The imagined graph before the second move of Red.}
    \label{f6}
 \end{figure}

 Suppose that Red has a strategy $S$ to win the game. After Red plays $vy$, the graph of the game has one star on three edges where two of them are red and one is blue. We denote the vertices as depicted in Figure \ref{f6sfig1}. Before his next move, Blue imagines that he has already claimed the edge $vj$, where $j$ is a black vertex, and that Red has not claimed the edge $vy$, see Figure \ref{f6sfig2}. Note that the edge $vj$ will remain free throughout the game, as otherwise Red would create a ${\cal CC}_{>3}$ in his graph.

 The imagined graph is isomorphic to the graph where the roles of the players are swapped. Blue imagines that he is the first player, and that Red claims the edge $vy$ as his second move. From now on Blue responds as advised by the winning strategy $S$ and wins the game, a contradiction.

\medskip
{\bf Case 3.} Vertex $x$ is pure blue and vertex $y$ is black, i.e. $x=i$.

  Suppose that Red has a strategy to win the game. After Red plays $iy$, the graph of the game has one $P_4$ which two non-adjacent edges are red, and the third one is blue. We denote the vertices as depicted in Figure \ref{f7sfig1}. Before his next move, Blue imagines that he has already claimed the edge $uy$, and that Red has not claimed the edge $iy$, see Figure \ref{f7sfig2}. Note that the edge $uy$ will remain free throughout the game, as otherwise Red would create a ${\cal CC}_{>3}$ in his graph.

 \begin{figure}[ht]
    \begin{subfigure}{.5\textwidth}
    \centering
     \includegraphics[width=0.6\linewidth]{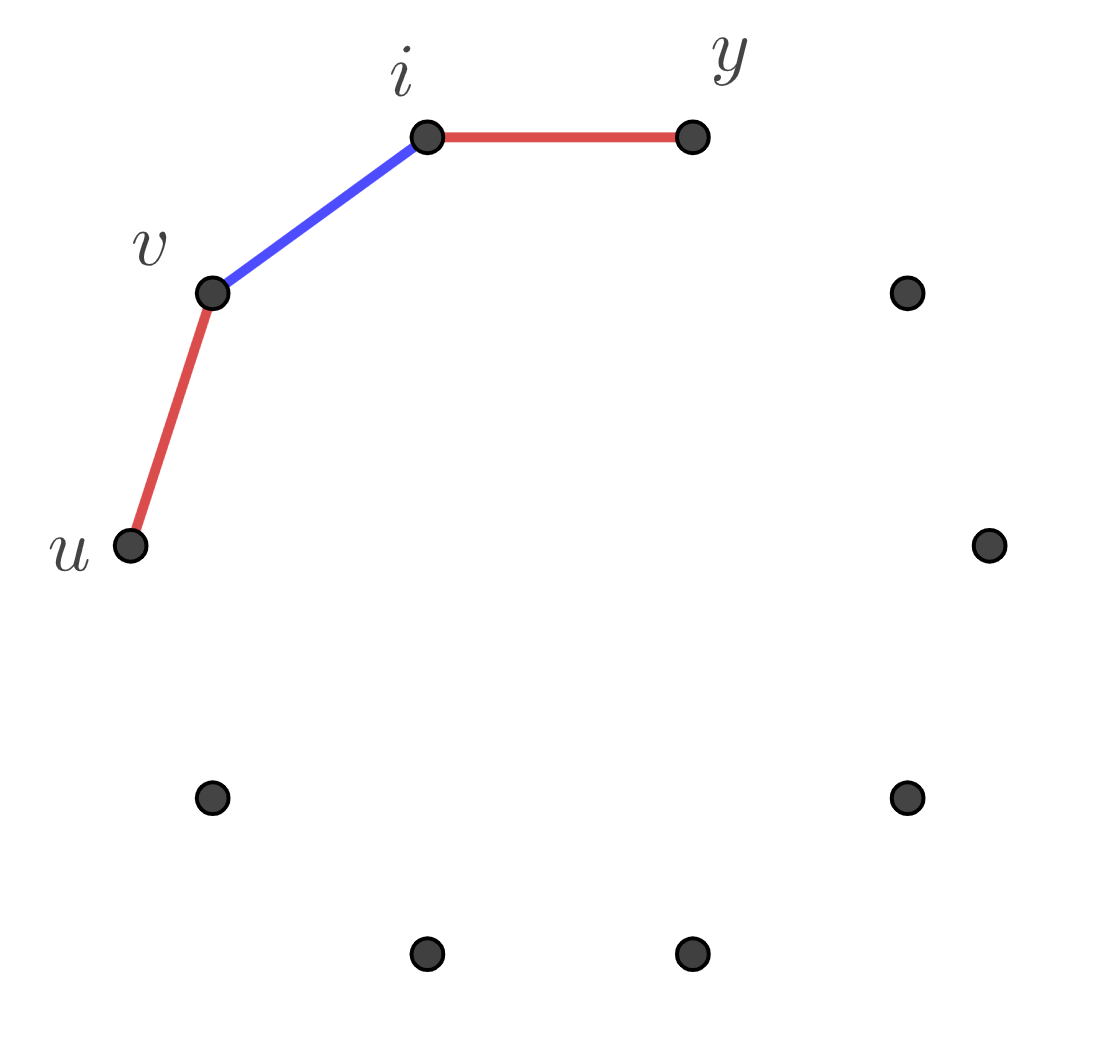}
     \caption{}
     \label{f7sfig1}
    \end{subfigure}%
    \begin{subfigure}{.5\textwidth}
     \centering
     \includegraphics[width=0.6\linewidth]{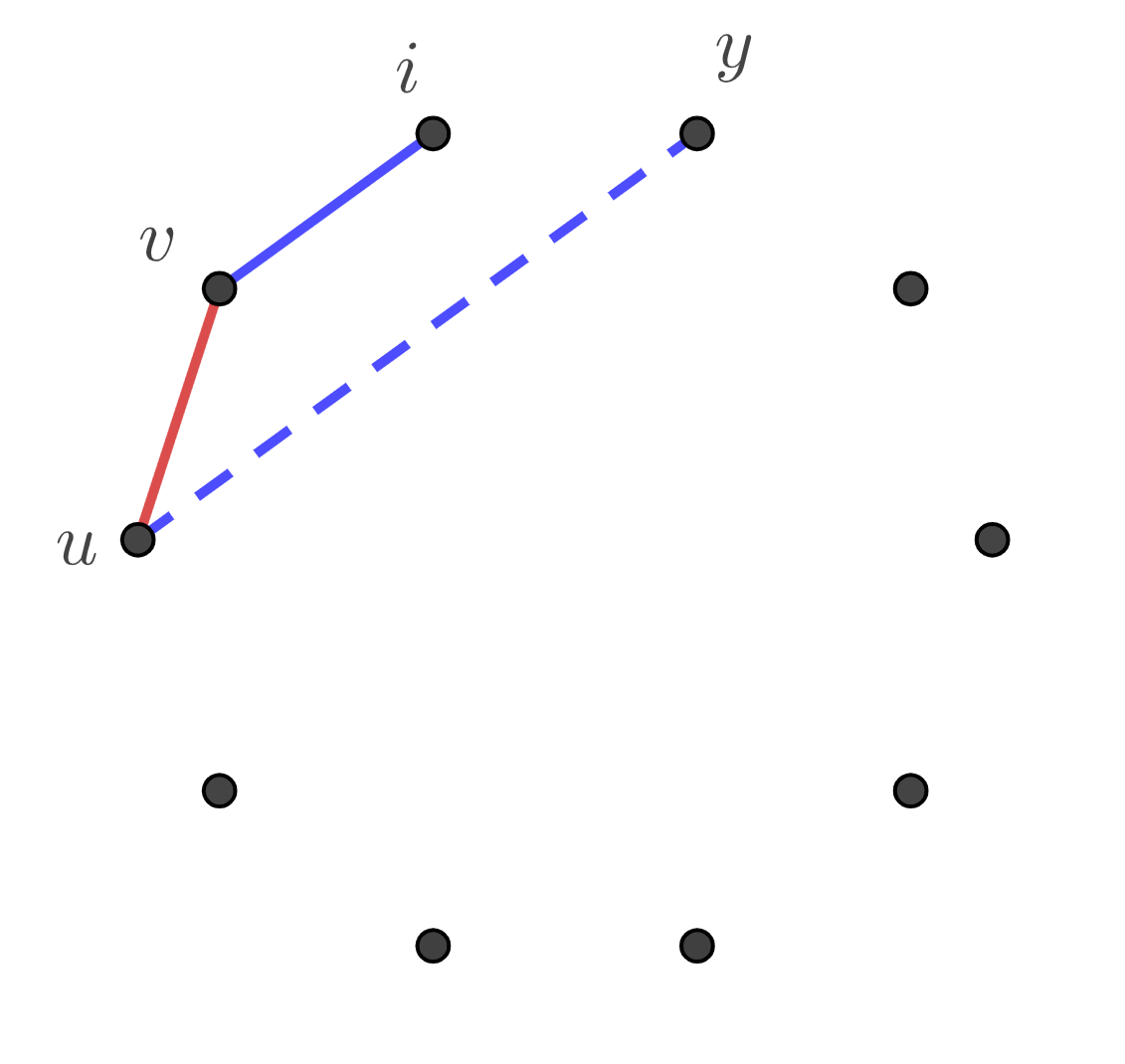}
    \caption{}
    \label{f7sfig2}
    \end{subfigure}
    \captionsetup{justification=centering}
    \caption{Case 3: (a) the graph before the second move of Blue. (b) The imagined graph before the second move of Red.}
    \label{f7}
    \end{figure}

 The imagined graph is isomorphic to the graph where the roles of the players are swapped. Blue imagines that he is the first player, and that Red claims the edge $iy$ as his second move. Hereafter, Blue responds as advised by the winning strategy $S$ and wins the game, a contradiction.

\medskip
\newpage
 {\bf Case 4.} Both $x$ and $y$ are black.

 Suppose that Red has a strategy $S$ to win the game. After Red plays $xy$, the graph of the game has one $P_3$, where one edge is red and the other one blue, and one isolated red edge. We denote the vertices as depicted in Figure \ref{f8sfig1}. Before his next move, Blue imagines that he has already claimed the edge $ux$, and that Red has not claimed the edge $uv$, see Figure \ref{f8sfig2}. Note that the edge $ux$ will remain free throughout the game, as otherwise Red would create a ${\cal CC}_{>3}$ in his graph.

 \begin{figure}[ht]
    \begin{subfigure}{.5\textwidth}
    \centering
     \includegraphics[width=0.6\linewidth]{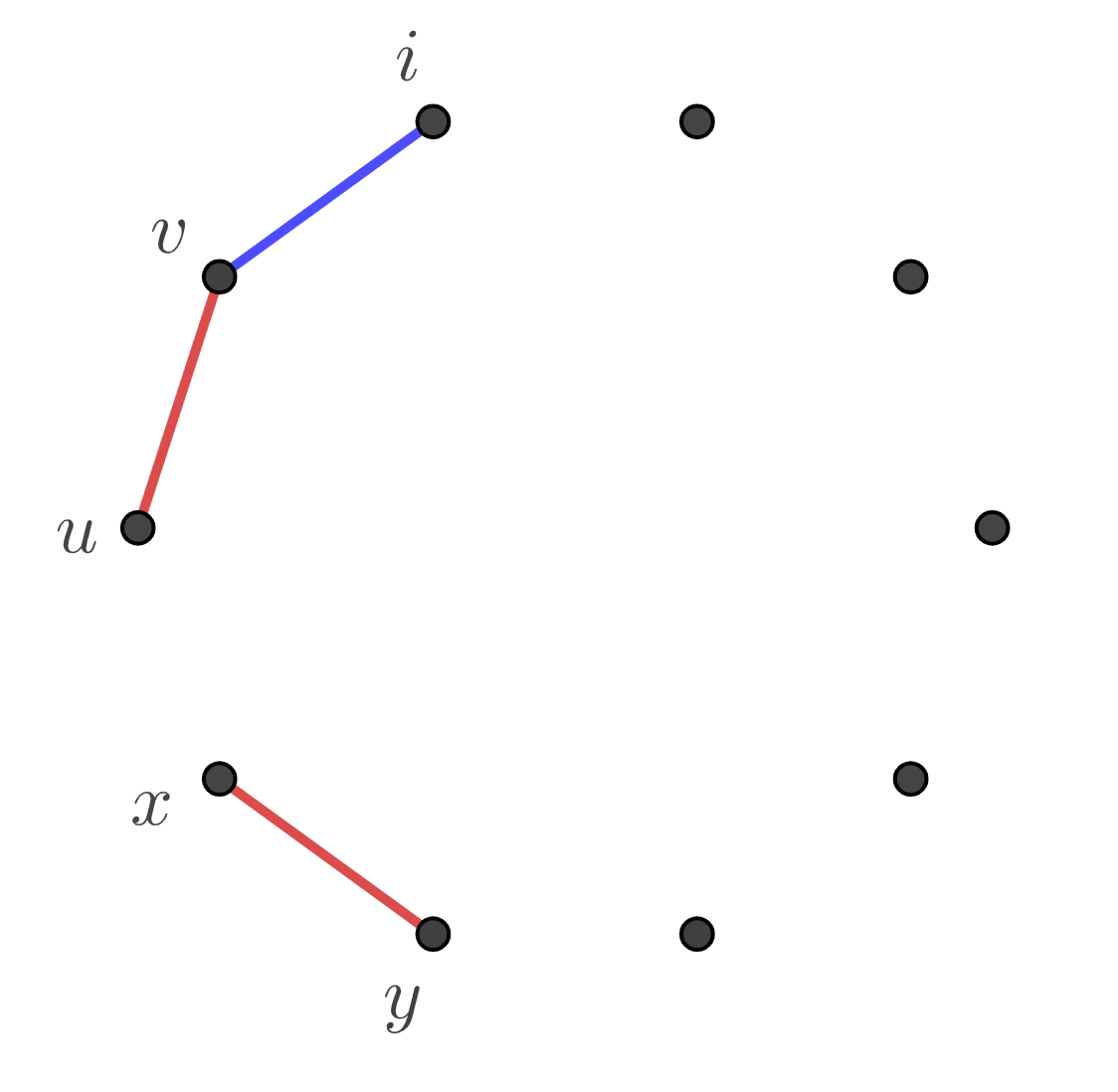}
     \caption{}
     \label{f8sfig1}
    \end{subfigure}%
    \begin{subfigure}{.5\textwidth}
     \centering
     \includegraphics[width=0.6\linewidth]{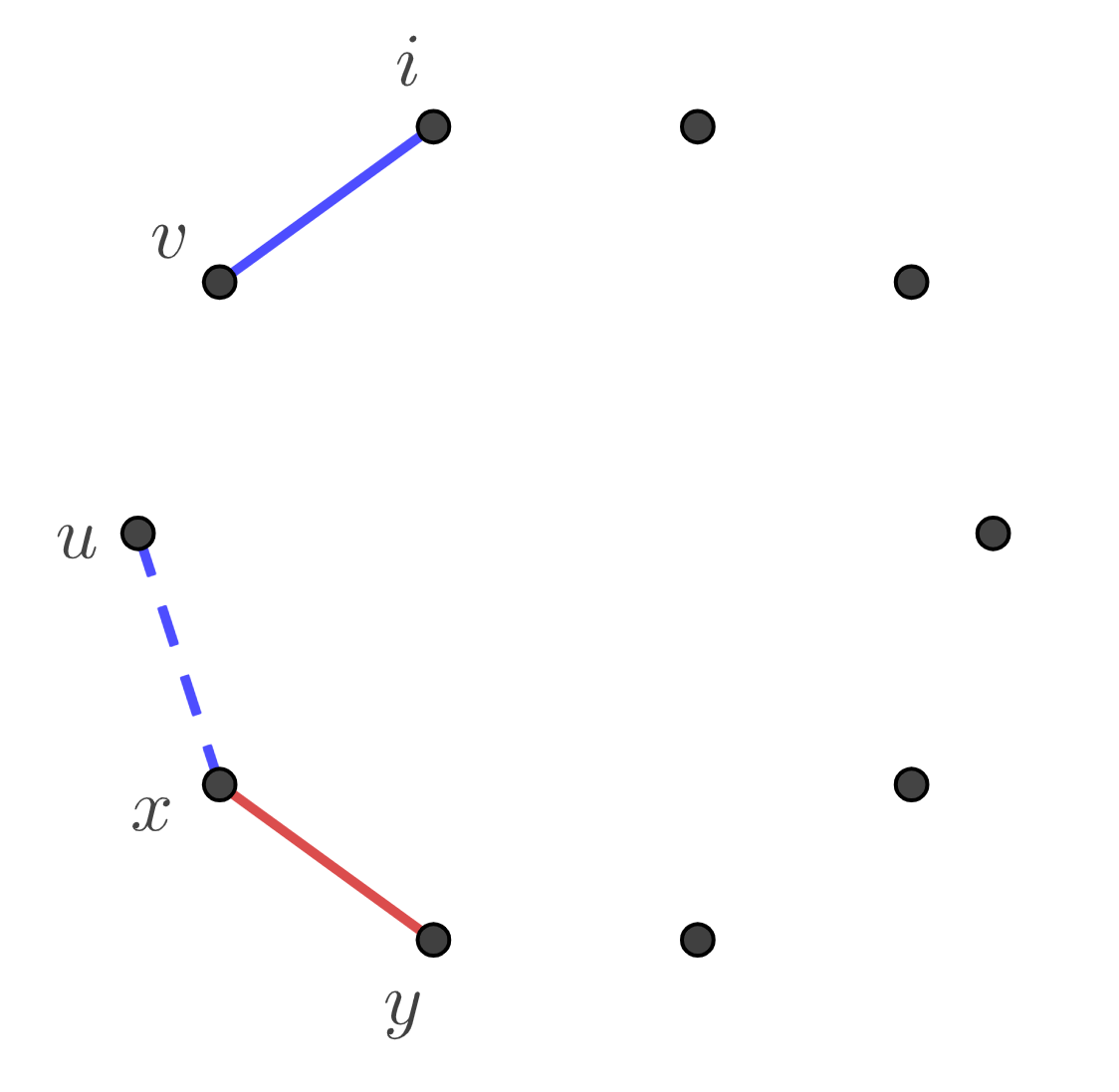}
    \caption{}
    \label{f8sfig2}
    \end{subfigure}
    \captionsetup{justification=centering}
    \caption{Case 4: (a) the graph before the second move of Blue. (b) The imagined graph before the second move of Red.}
    \label{f8}
    \end{figure}

 The imagined graph is isomorphic to the graph where the roles of the players are swapped. Blue imagines that he is the first player, and that Red claims the edge $uv$ as his second move. Hereafter, Blue responds as advised by the winning strategy $S$ and wins the game, a contradiction.

\medskip
 {\bf Case 5.}  Vertex $x$ is pure red and vertex $y$ is pure blue, i.e. $x=u$ and $y=i$.

 Blue is the second player, so his graph can never have more edges than the Red's graph. Having that in mind, as well as Observation \ref{o1}, we will describe an explicit strategy for Blue to claim disjoint triangles until the very end, which will enable him to win.
     Note that in this case the first $RC$ is a $P_3$, so using Observation \ref{o1}, it is not possible for Red to have more than $n-1$ edges if $n=3m$, or $n-2$ otherwise.

 Let us introduce some terminology. During the game, whenever a blue edge is added to a pure red $P_3$ so that it completes a triangle, we call it a \emph{nice edge}. Furthermore we call the pure red vertex that is incident with that triangle a \emph{nice vertex}.

 First we will give a strategy for Blue that he follows from his second move on, and afterwards we will show that Blue can follow it.

 \medskip
 {\bf Stage 1.} While there is at least one black vertex in the game, Blue repeatedly plays by the first rule in this list that is applicable. Let $k$ denote the number of nice edges in Blue's graph.
 \begin{enumerate}
     \item\label{1} If there is a pure red $P_3$,\\ Blue claims the isolated edge that completes a triangle when added to the pure red $P_3$. In other words, Blue claims a nice edge.

        \item\label{2a} If $k=0$,
            \begin{enumerate}
                \item\label{2ai} Pure red vertices are in at least three different $RC$.\\
                 Blue chooses two pure red vertices from two different $RC$, say $x$ and $y$ and claims an edge between them. We denote by $X_3$ the third $RC$.
                Then, Blue repeatedly plays by the first rule in the following list that is applicable until he claims all three edges of a new triangle $xyz$.
                \begin{enumerate}
                    \item If there is a pure red $P_3$ anywhere in the graph,\\ Blue claims the isolated edge that added to the pure red $P_3$ completes a triangle, i.e. a nice edge.
                    \item If degree of the vertex $x$ in Blue's graph is one,\\
                    Blue claims the edge $xz$, where $z$ is a pure red vertex in $X_3$.
                    \item Otherwise (the degree of the vertex $x$ in Blue's graph is two),\\
                    Blue claims the edge $yz$, where $y$ and $z$ are endpoints of the blue edges incident to $x$.
                \end{enumerate}

                \item\label{2aii} The pure red vertices are in exactly two $RC$.\\
                Denote by $X_1$ the $RC$ that has two pure red vertices, and the other one with $X_2$, later we will see that $X_1$ exists. Blue claims an edge between pure red vertex from $X_2$ and one black vertex, denote it by $w$.
                Then, Blue repeatedly plays by the first rule in the following list that is applicable until he claims all three edges of a new triangle.
                \begin{enumerate}
                    \item If there is a pure red $P_3$ anywhere in the graph,\\ Blue claims the isolated edge that added to the pure red $P_3$ completes a triangle, i.e. a nice edge.
                    \item If the degree of the vertex $w$ in Blue's graph is one,\\
                    Blue claims an edge between a pure red vertex from $X_1$ and $w$.
                    \item Otherwise (the degree of the vertex $w$ in Blue's graph is two),\\
                    Blue claims the edge that completes the triangle.
                \end{enumerate}

            \end{enumerate}

        \item\label{2b} If $k=1$,\\
        Blue claims an edge adjacent to the nice edge and incident with a black vertex.
        Then, Blue repeatedly plays by the first rule in the following list that is applicable until he claims a new blue triangle.
                \begin{enumerate}
                    \item If there is a pure red $P_3$ anywhere in the graph,\\ Blue claims the isolated edge that added to the pure red $P_3$ completes a triangle, i.e. a nice edge.
                    \item Otherwise,\\
                    Blue claims the edge that completes the blue triangle.
                \end{enumerate}

        \item\label{2c} If $k>1$,\\
        Blue chooses a nice vertex and then a nice edge that are not in the same $RC$, and then claims an edge between them.
        Then, Blue repeatedly plays by the first rule in the following list that is applicable until he claims a new blue triangle.
                \begin{enumerate}
                    \item If there is a pure red $P_3$ anywhere in the graph,\\ Blue claims the isolated edge that added to the pure red $P_3$ completes a triangle, i.e. a nice edge.
                    \item Otherwise,\\
                    Blue claims the edge that completes the blue triangle.
                \end{enumerate}
 \end{enumerate}

 Now we will prove that when it is Blue's turn to play Stage 1, he can follow it.

 When it is Blue's turn to play by rule \ref{1}, it is clear that he can claim that edge.
 Since the strategy of Blue is to make disjoint triangles and nice edges, before he plays any of the rules, his graph consists of $t$ disjoint triangles, $k$ nice edges and isolated vertices, where $t$ or $k$ are possibly zero. Obviously, these isolated vertices have to be black or pure red.

 \begin{cl}\label{c1}
     When it is Blue's turn to play and his graph consists of disjoint triangles and isolated vertices, there are at least three pure red vertices.
    \end{cl}
     \begin{proof}
     Blue's graph is a disjoint union of $t$ triangles and isolated vertices, so it has $3t$ blue vertices and $3t$ edges. At this moment Red's graph has $3t+1$ edges. The first $RC$ is a $P_3$ and the Red's graph without that component has $3t-1$ edges. The extremal graph that is described in Observation \ref{o1} gives the smallest number of vertices for a fixed number of edges. Therefore, a ${\cal CC}_{>3}$-free graph with $3t-1$ edges has at least $3t$ vertices. Then the Red's graph has at least $3t+3$ red vertices three more than the number of blue vertices, so at least three of them must be pure red.
     \end{proof}

    \begin{cl}\label{c2}
   When it is Blue's turn to start playing by rule 2 (his graph consists of disjoint triangles and isolated vertices), there are at most two pure red vertices in any $RC$.
    \end{cl}
     \begin{proof}
     Suppose there is a $RC$ with three pure red vertices.\\
     It cannot be a pure red triangle, otherwise one move before Red made a pure red triangle, he had to have a pure red $P_3$. Realising that just one pure red $P_3$ can be made per move, Blue responds by claiming the nice edge following his strategy.\\
     Therefore, the $RC$ has to be a pure red $P_3$, but than rule \ref{1} would be activated and Blue would claim a nice edge and that leads to a contradiction.
     \end{proof}

 Note that Red cannot claim any edge adjacent to a nice edge or incident to a nice vertex, also he cannot claim any edge between two $RC$, otherwise he would create a ${\cal CC}_{>3}$.
 Every $RC$ has to be an isolated edge, a $P_3$ or a triangle, by Observation \ref{o_1}.

 Now we will show that at the moment when Blue has to play by one of the rules \ref{2a}, \ref{2b} or \ref{2c} he can do that, and in particular he claims a new triangle.

 \begin{enumerate}
     \item  $k=0$. \\
     Note that in this moment there are at least three pure red vertices by Claim \ref{c1}. If Blue claimed a nice edge in the middle of rule \ref{2a}, in $RC$ where Blue has made a nice edge, there is still one pure red vertex, precisely a nice vertex.
     \begin{enumerate}
         \item  There are pure red vertices in at least three $RC$.\\
         It is clear that Blue can claim the edge $xy$.
         When it is Blue's turn to claim the edge $xz$, it is available for him because, as mentioned above, $z$ must exists even if Blue has made a nice edge in the meantime. Clearly, Blue can follow rule \ref{2a} to the end and claim the triangle $xyz$.
         \item  The pure red vertices are in two $RC$.\\
         We know that in one $RC$ there can be at most two pure red vertices by Claim \ref{c2}. Also, there are at least three of them as proven in Claim \ref{c1}, therefore one $RC$ has to have precisely two pure red vertices, so $X_1$ exists. Then $X_2$ has at least one pure red vertex and it is clear that Blue can follow rule \ref{2a} to the end and claim a new triangle, by same reasoning as above.

     \end{enumerate}

     Note that it is not possible that all pure red vertices are in one $RC$ because of Claims \ref{c1} and \ref{c2}.

     \item  $k=1$. Clearly, Blue can follow rule \ref{2b}.

     \item  $k>1$.
     First we prove the following claim needed to complete the proof for Stage 1.

     \begin{cl}\label{c3}
    If there are $k>1$ nice edges, then there are at least $k-1$ nice vertices.
    \end{cl}
     \begin{proof}
     We will go through the whole strategy of Blue to determine when the numbers of nice vertices and nice edges are changing.\\
     In rule \ref{1} both numbers increase by one.\\
     In rule \ref{2a} it can happen that just the number of nice edges increases by one, that both numbers increase by one or that there is no change.\\
     In rule \ref{2b} the number of nice edges decreases by one.\\
     In rule \ref{2c} both numbers decrease by one.\\
     Therefore the only way to make the number of nice vertices less than the number of nice edges for one is using rule \ref{2a}. That can happen only once because Blue will not use rule \ref{2a} again as long as $k \neq 0$.
     \end{proof}
     Claim \ref{c3} ensures that Blue can follow rule \ref{2c} to the end and claim a new triangle.\end{enumerate}

Note that during Stage 1, Blue's graph consists of $t$ disjoint triangles, $k$ nice edges and some isolated vertices.

 When there are no more moves in Stage 1, clearly there are no more black vertices in the game.

\medskip
{\bf Stage 2.} In this stage we keep the structure and all the rules from Stage 1, we just change the following two rules:
\begin{itemize}
    \item[(\ref{2aii})]  If the pure red vertices are in exactly two $RC$ and $k=0$,\\
    Blue claims a $P_3$ taking just pure red vertices.

    \item[\ref{2b}.]  If $k=1$,\\
    Blue claims an edge adjacent to the nice edge and incident with a pure red vertex, that is not the nice vertex from the same $RC$. Then, in the following move he claims the edge that together with the nice edge and the edge that he just claimed closes a triangle.
\end{itemize}

We will prove that when Blue can no longer follow his strategy, he has already won.
Note that at that point every isolated vertex in Blue's graph is pure red.

When the Blue's strategy tells him to play by rule \ref{1}, that means that in his last move Red joined an isolated edge and a black vertex and made a pure red $P_3$. Therefore, this move can only be the first move of Blue in Stage 2, because at that point there are no more black vertices. Therefore, we conclude that it is not possible for Red to claim a pure red triangle in this stage, and Blue will never claim a nice edge while performing the rules \ref{2a}, \ref{2b} or \ref{2c}.

 \begin{cl}\label{c4}
    As long as Red has not already lost, Blue can play by one of the rules \ref{2a}, \ref{2b} or \ref{2c}.
\end{cl}
\begin{proof}

     When the Blue's strategy tells him to play by rule \ref{2ai}, \\
     Blue can claim a new triangle in the following three moves, the argument is the same as in Stage 1. After that, his graph will have the same structure as in the beginning of Stage 2 and he continues to play.

     When the Blue's strategy tells him to play by rule \ref{2aii}, \\
     note that it is not possible to have less than three pure red vertices here, because in that case Blue would have already won with at most $n-2$ edges in his graph, $3m+1 \leq n \leq 3m+2$, for some integer $m$ (Blue's graph consists of triangles and isolated vertices). Red's graph cannot have more than $n-2$ edges, because of his first component and Observation \ref{o1}.

    These pure red vertices are the only vertices that are not blue, and by Claim \ref{c2} we know that in two $RC$ we have at most four pure red vertices. Clearly, Blue can claim a $P_3$ in the following two moves. Now, we show that after these moves Blue wins.
    \begin{itemize}
         \item If there are four pure red vertices,
         we know that $n=3m+1$, for some integer $m$, and after making a $P_3$ Blue's graph has $n-2$ edges. At the same time Red's graph cannot have more than $n-2$ edges, because of his first component and Observation \ref{o1}, so Blue wins.

        \item If there are three pure red vertices.
         We know that $n=3k$, for some integer $k$, and after making a $P_3$ Blue's graph has $n-1$ edges. At the same time Red's graph cannot have more than $n-1$ edges, because of his first component and Observation \ref{o1}, so Blue wins.
    \end{itemize}
     \vspace{0.2cm}

     When Blue's strategy tells him to play by rule \ref{2b}, \\
    there is exactly one nice edge and if there is a pure red vertex, that is not the nice vertex from the same $RC$, in his following two moves Blue can claim a triangle as advised by the strategy. Then, his graph has the same structure as in the beginning of Stage 2 and he continues to play.

    Otherwise, if the only vertex that is not blue is the nice vertex from the same $RC$, then Blue has already won. This is true because $n=3m$, for some integer $m$, and the Blue's graph has $n-2$ edges ($3t+1$, where $t$ is the number of blue triangles), while Red's graph cannot have more edges because it contains at least two $P_3$.

    Note that it is not possible that there is no pure red vertices, because in that case Blue's graph would have $n-1$ edges, and Red's graph cannot have more then $n-3$ edges, for the same reason as above ($n=3m+2$, for some integer $m$).

    When Blue's strategy tells him to play by rule \ref{2c}, \\
    the argument is the same as in Stage 1, and in his following two moves Blue can claim a triangle. After that, his graph has the same structure as in the beginning of Stage 2 and he continues to play.
\end{proof}

This concludes the proof of Case 5, and thus also the proof of the theorem.
\qedsymbol

\section{Strong CAvoider-CAvoider games} \label{s:caca}

\vspace*{3mm}

{\bf Proof of Theorem \ref{TH004}:}
In order not to lose each player must keep the maximal degree in his graph at most two. Furthermore, the rules of the game dictate that both players must maintain their respective graphs connected throughout the game. Hence, as long as no one loses the game, the graph of each player must be a path or a cycle.
Note that if Blue can claim a Hamiltonian cycle he will win, because Red will lose in the following move.

We will show that Blue can win.
 In the beginning, we have a graph $G$ with $n$ isolated vertices, and Red claims an edge, let us denote it by $uv$. Then Blue claims an edge that is not adjacent to the red one, let us denote it by $rt$. In the following move Red has two options, up to isomorphism, for choosing an edge $e=xy$, which will be our two cases.

\medskip
 \noindent{\bf Case 1.} Vertex $x$ is red and $y$ is blue, w.l.o.g. $x=u$ and $y=t$.

  We again apply the strategy stealing argument, assuming that after his second move Red has a strategy $S$ to win the game.

After Red plays $ut$ it is Blue's turn. The graph of the game consists of one $P_4$, with two adjacent red edges and one blue edge, and isolated vertices, see Figure \ref{f9sfig1}. Before his next move, Blue imagines that he has already claimed the edge $ur$, and that Red has not claimed the edge $ut$, see Figure \ref{f9sfig2}.
Note that the edge $ur$ will remain free throughout the game, as otherwise Red would create a $S_3$ in his graph.

 \begin{figure}[ht]
    \begin{subfigure}{.5\textwidth}
    \centering
     \includegraphics[width=0.6\linewidth]{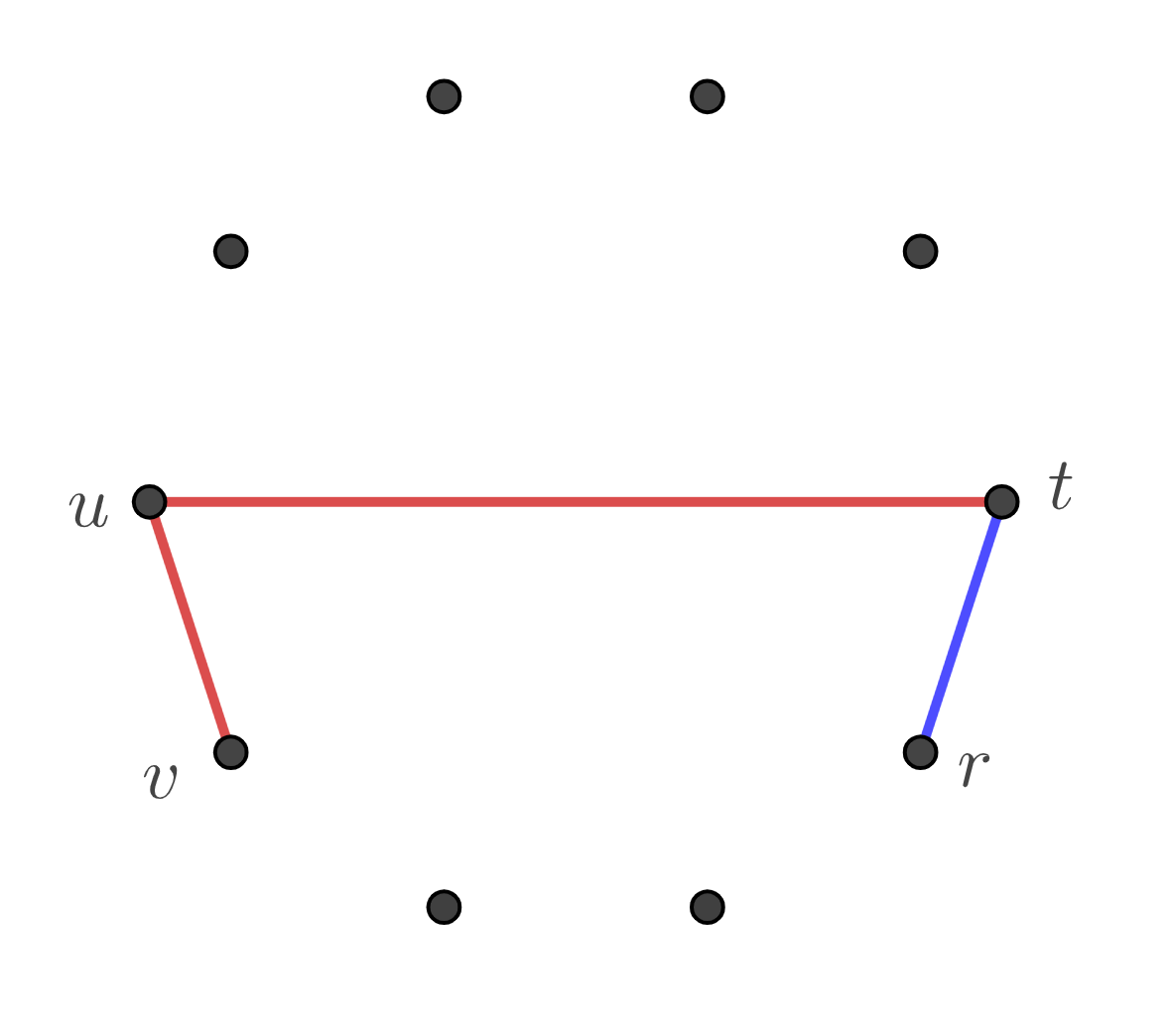}
     \caption{}
     \label{f9sfig1}
    \end{subfigure}%
    \begin{subfigure}{.5\textwidth}
     \centering
     \includegraphics[width=0.6\linewidth]{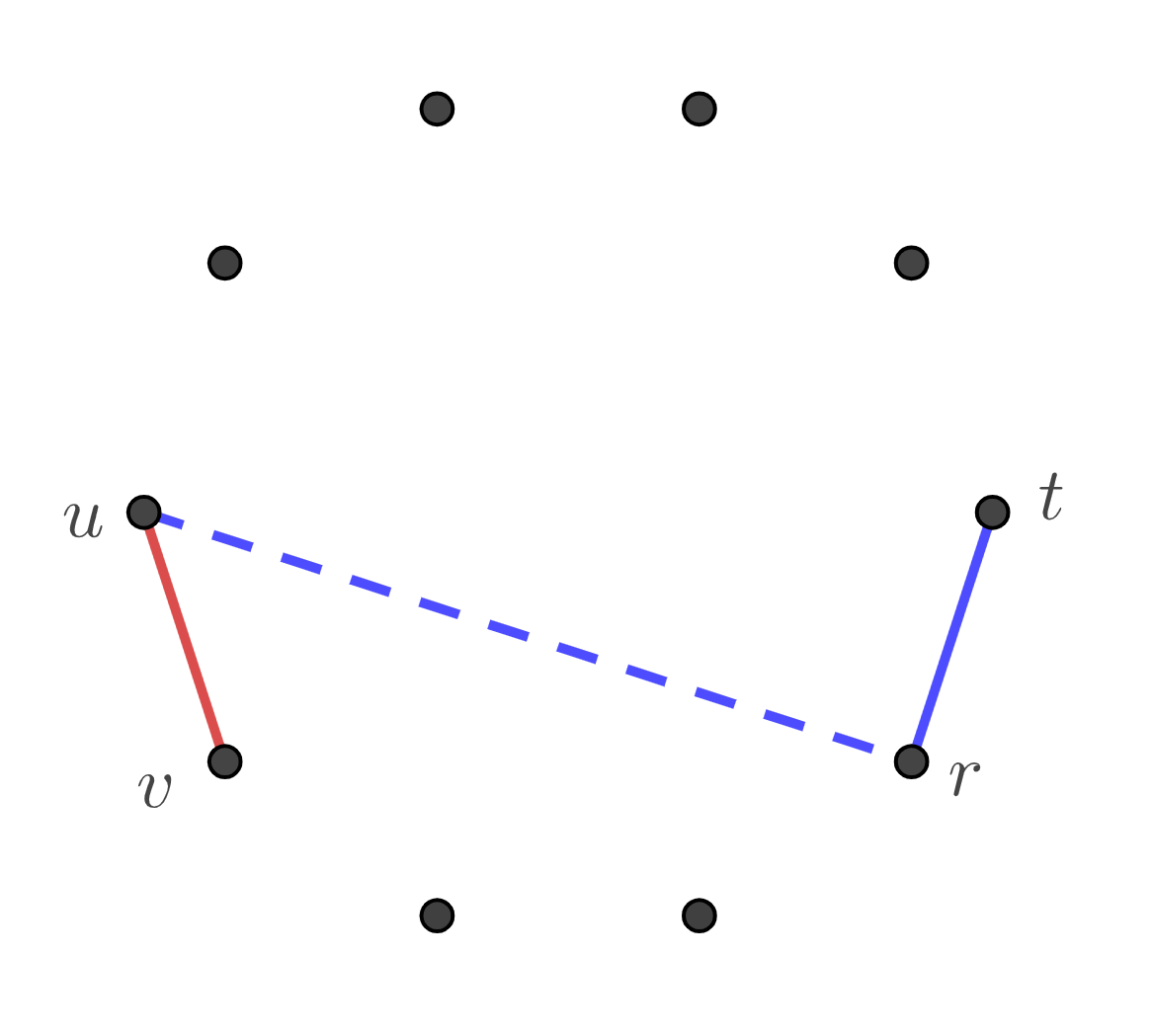}
    \caption{}
    \label{f9sfig2}
    \end{subfigure}
    \captionsetup{justification=centering}
    \caption{Case 1: (a) the graph before the second move of Blue. (b) The imagined graph before the second move of Red.}
    \label{f9}
    \end{figure}

 The imagined graph is isomorphic to the graph where the roles of the players are swapped. Blue imagines that he is the first player, and that Red claims the edge $ut$ as his second move. From now on, Blue responds as advised by the strategy $S$. Because this is a winning strategy, Blue wins the game, a contradiction.

\medskip
 \noindent{\bf Case 2.} Vertex $x$ is red and $y$ is black, w.l.o.g. $x=u$.

 We will first describe a strategy for Blue, and then we will show that he can follow it and win.
  The following move of Blue is the edge $tv$. In his third move, if it is unclaimed, Blue claims the edge $vy$, see Figure \ref{f10sfig1}.


  \begin{figure}[ht]
    \begin{subfigure}{.5\textwidth}
    \centering
     \includegraphics[width=0.6\linewidth]{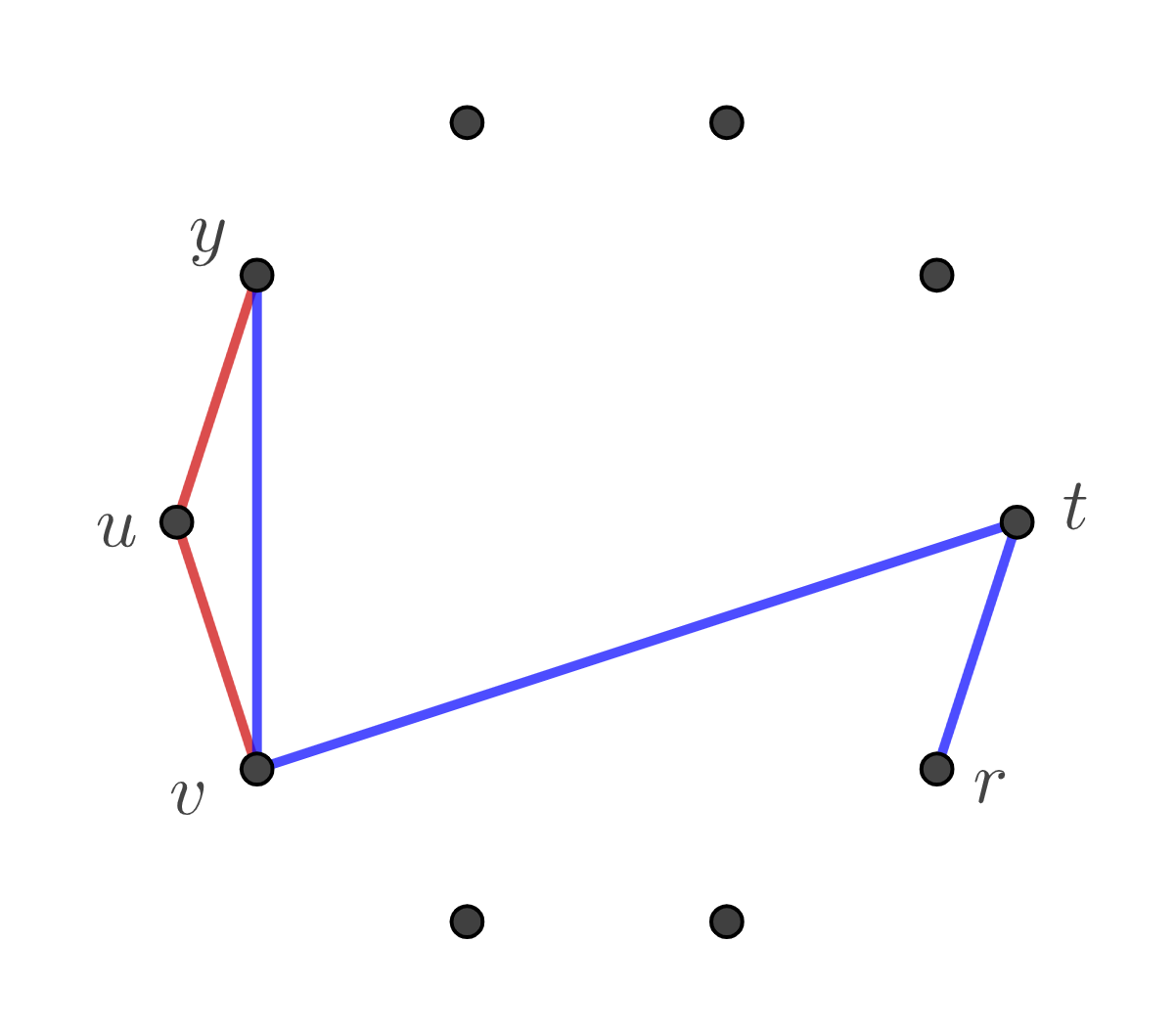}
     \caption{}
     \label{f10sfig1}
    \end{subfigure}%
    \begin{subfigure}{.5\textwidth}
     \centering
     \includegraphics[width=0.6\linewidth]{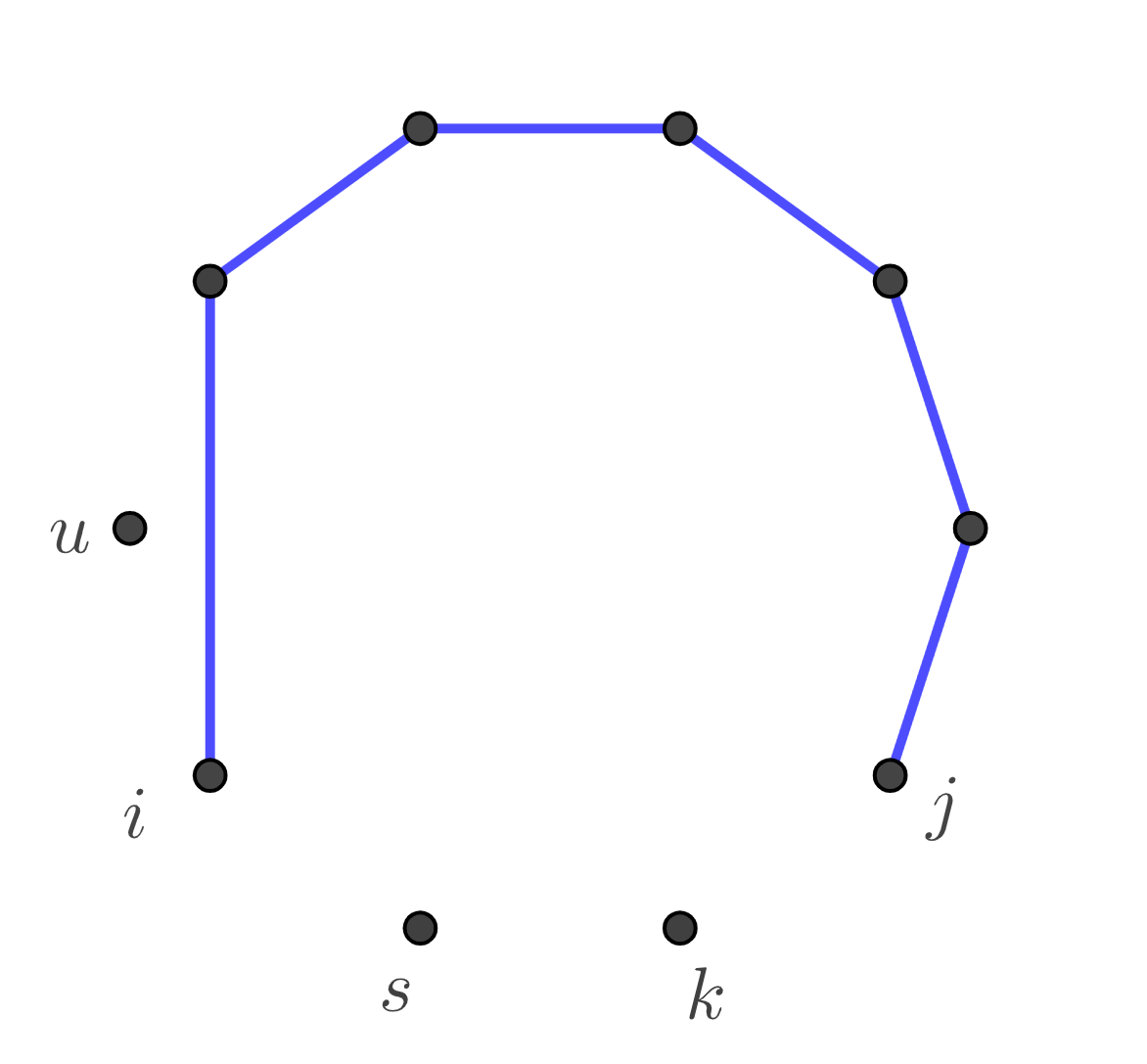}
    \caption{}
    \label{f10sfig2}
    \end{subfigure}
    \caption{Case 2: (a) Blue's graph after his third move. (b) Blue's graph after his $(n-4)th$ move.}
    \label{f10}
    \end{figure}




     If Blue's graph is not a path on $n-3$ vertices, then in his following move, Blue claims an edge incident with $y$ and one black vertex. From that point on, the strategy of Blue will be to make a Hamiltonian path on the vertex set $V \backslash \{u\}$ and then to complete it to a Hamiltonian cycle by connecting the vertex $u$ to both ends of the blue path.

     While there are more than two vertices in $V \backslash \{u\}$ that are not in the blue path, Blue extends his path by adding one of the vertices from $V \backslash \{u\}$.\\
     Then, we denote the last two isolated vertices besides $u$ in the Blue's graph with $s$ and $k$, and the ends of the Blue's path with $i$ and $j$, see Figure \ref{f10sfig2}.

     We distinguish these cases:
     \begin{enumerate}
         \item[2.1]\label{4c1} If just one of the edges $\{is, ik, js, jk\}$ is free, or two of them are, but both red edges are incident with one of the vertices $\{s, k\}$, see Figure \ref{f11sfig1}, \ref{f11sfig2},\\
         Blue will claim a free edge from the set $\{is, ik, js, jk\}$, and then the edge $sk$.
         \item[2.2]\label{4c2} If two of the edges $\{is, ik, js, jk\}$ are red and exactly one of those edges is incident with vertex $s$, see Figure \ref{f11sfig3} and \ref{f11sfig4},\\
         Blue will claim a free edge from the set $\{is, ik, js, jk\}$, and then the edge $sk$ if it is free, otherwise the remaining one from the set of edges $\{is, ik, js, jk\}$.
         \item[2.3]\label{4c3} If exactly one of the edges $\{is, ik, js, jk\}$ is red,\\
          w.l.o.g. let us assume that the edge $is$ is red, see Figure \ref{f11sfig5}. Blue claims the edge $ik$, and then in the following move he claims one of the edges $\{ks, js\}$.
         \item[2.4]\label{4c4} If there are no red edges in $\{is, ik, js, jk\}$, we have two subcases.
         \begin{enumerate}
             \item If the edge $sk$ is free,\\
             Blue claims the edge $ik$. In his following move he claims one of the edges $\{ks, js\}$.
             \item  If the edge $sk$ is red,\\
             at least one of the vertices $\{s, k\}$ has to have degree two in the Red's graph, w.l.o.g. let us assume that vertex is $k$.\\
             Then, Blue claims the edge $is$ and in the following move he claims the edge $jk$, see Figure \ref{f11sfig6}.
         \end{enumerate}

         \item[2.5]\label{4c5} Else, Blue claims any free edge that does not make an $S_3$ in his graph.
      \end{enumerate}

    \begin{figure}[ht]
    \begin{subfigure}{.333\textwidth}
    \centering
     \includegraphics[width=0.8\linewidth]{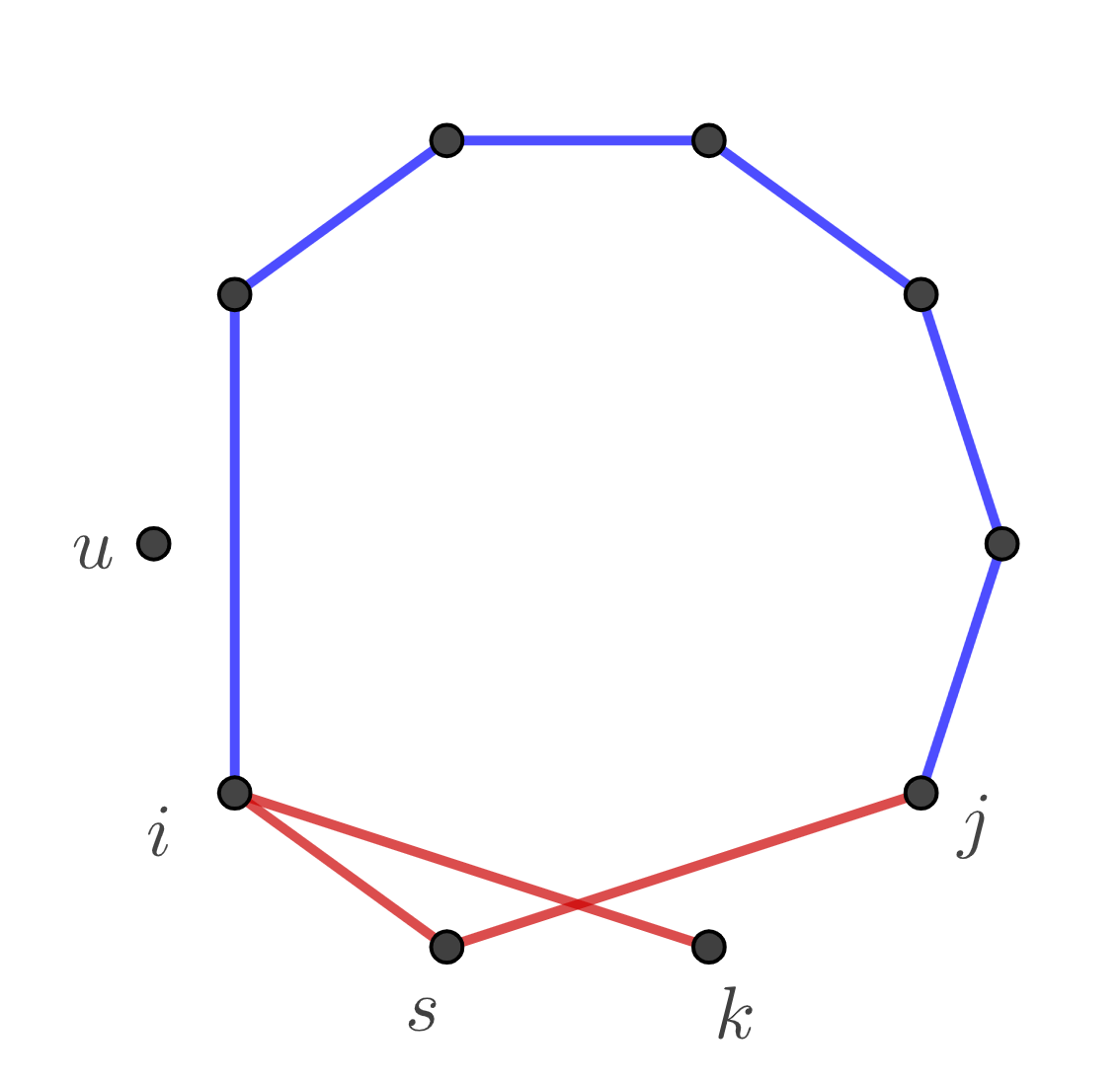}
     \caption{}
     \label{f11sfig1}
    \end{subfigure}%
    \begin{subfigure}{.333\textwidth}
     \centering
     \includegraphics[width=0.8\linewidth]{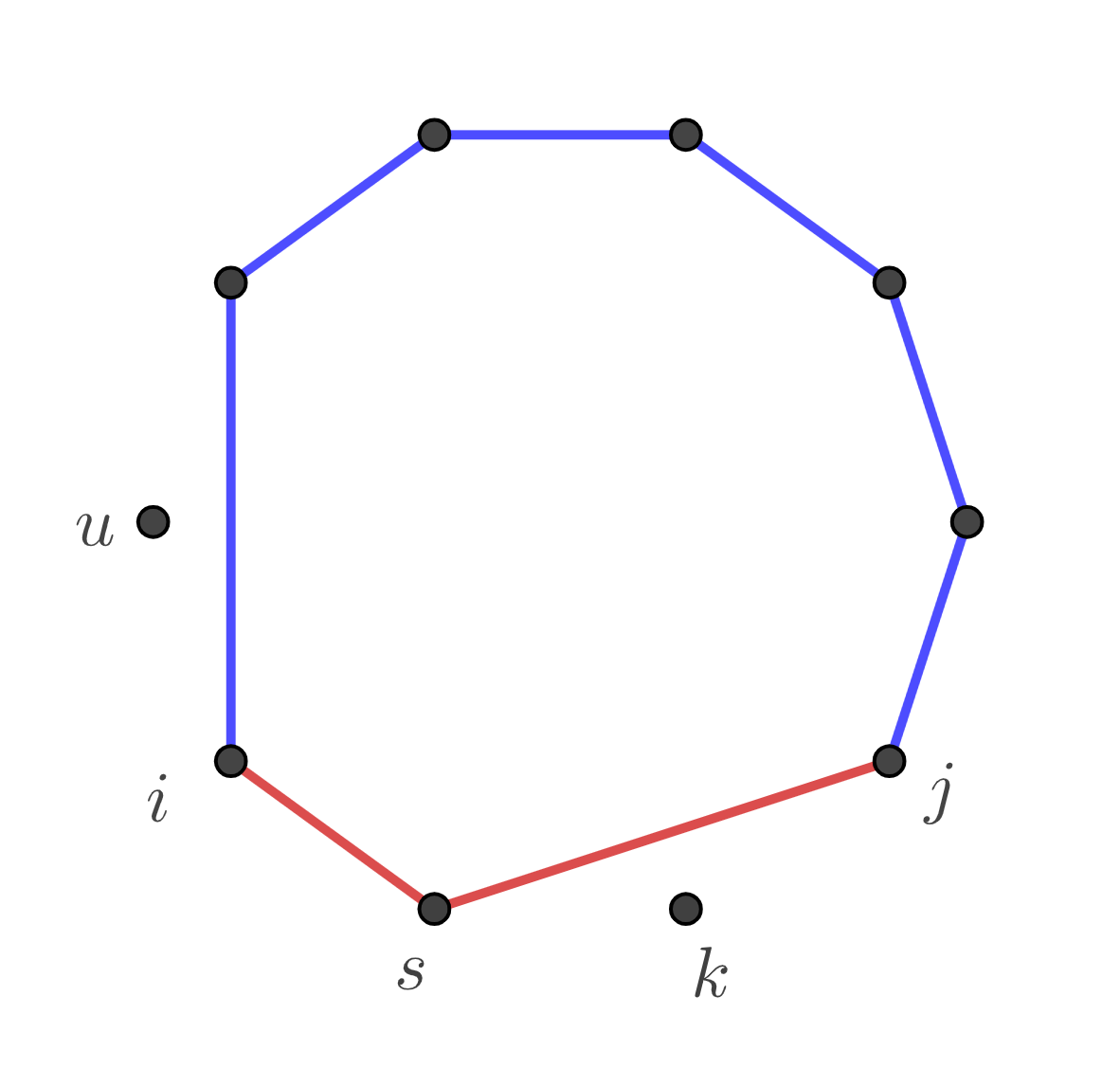}
    \caption{}
    \label{f11sfig2}
    \end{subfigure}
    \begin{subfigure}{.333\textwidth}
    \centering
     \includegraphics[width=0.8\linewidth]{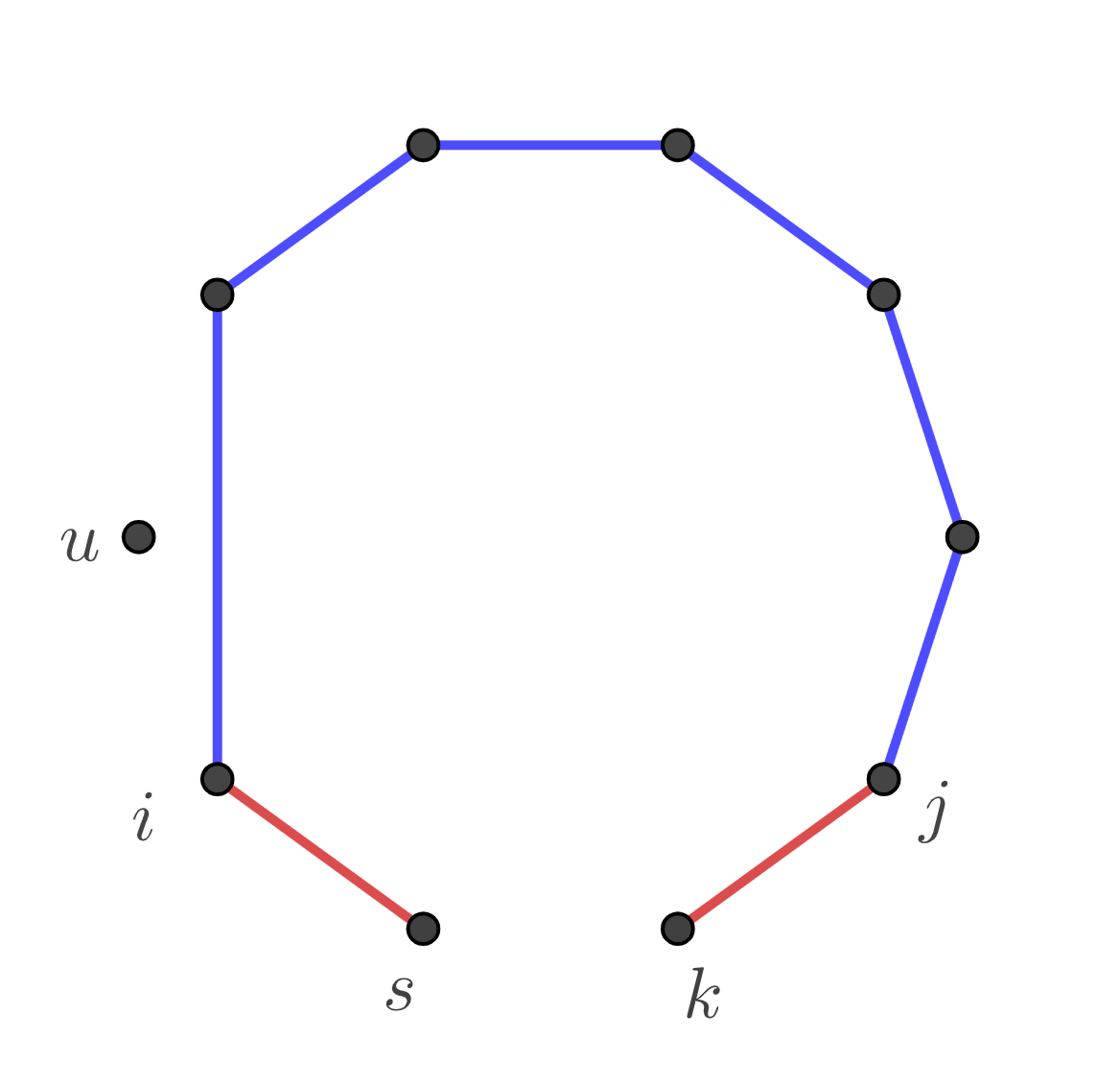}
     \caption{}
     \label{f11sfig3}
    \end{subfigure}\\
    \begin{subfigure}{.333\textwidth}
    \centering
     \includegraphics[width=0.8\linewidth]{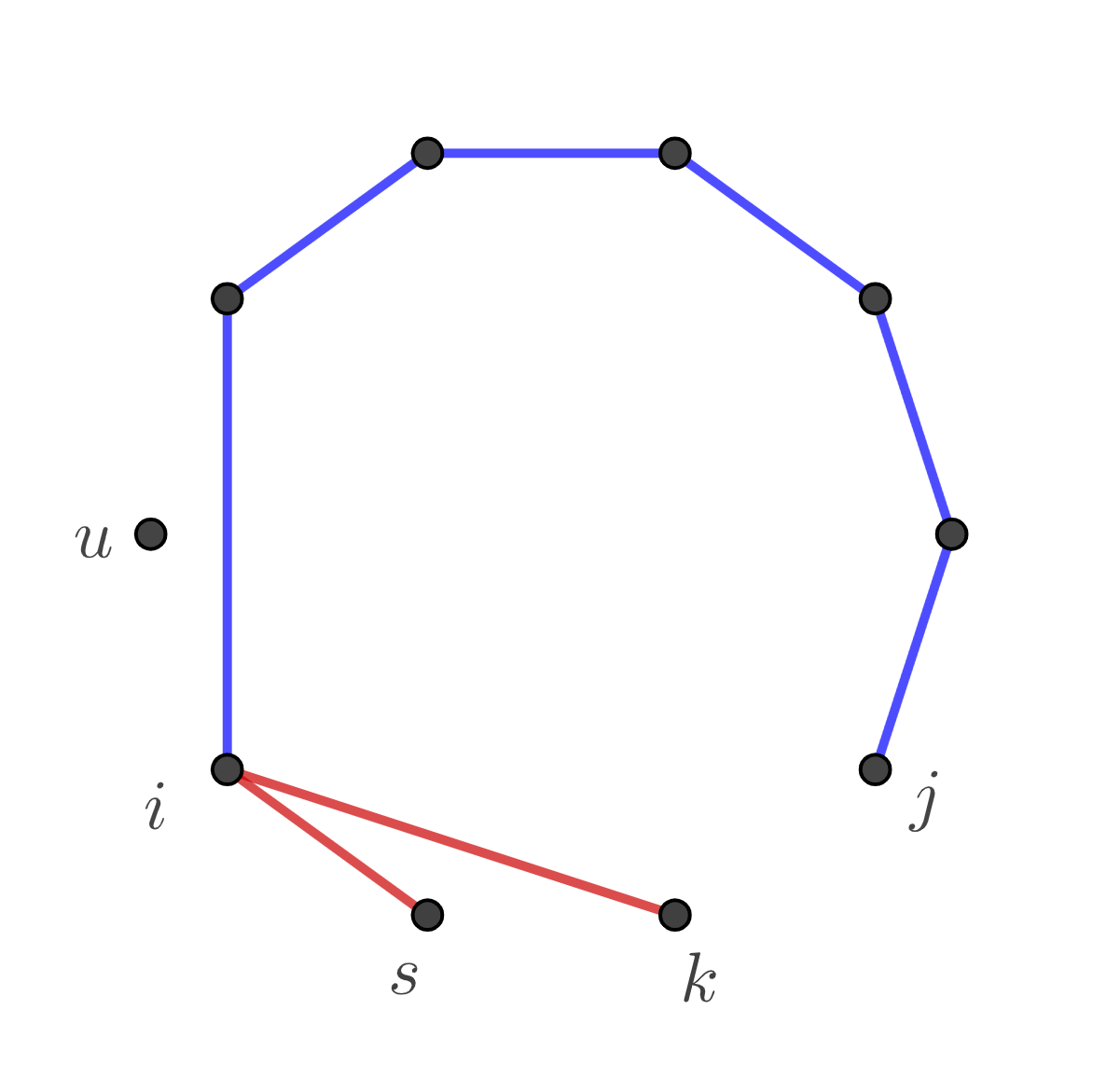}
     \caption{}
     \label{f11sfig4}
    \end{subfigure}
    \begin{subfigure}{.333\textwidth}
    \centering
     \includegraphics[width=0.8\linewidth]{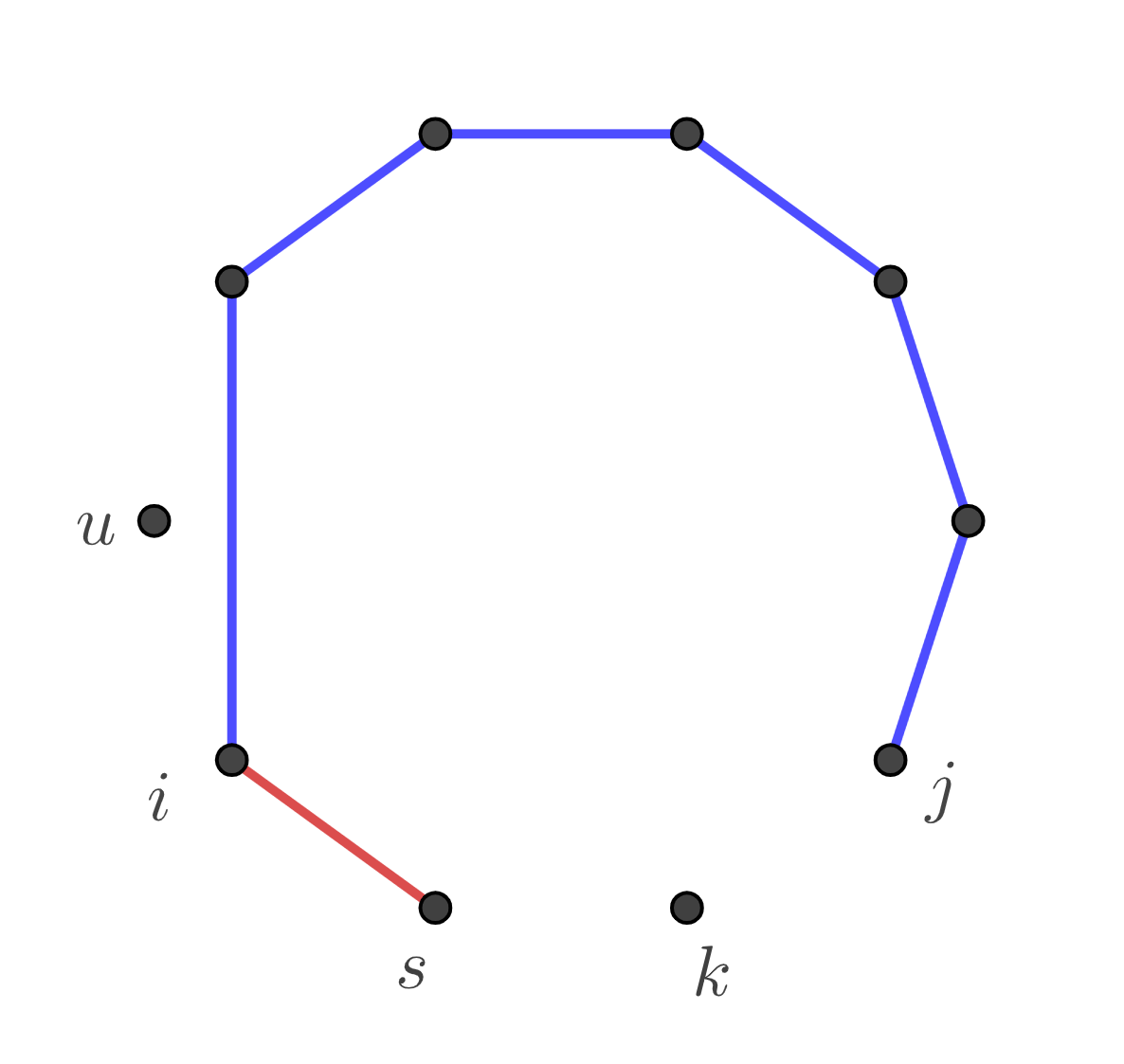}
     \caption{}
     \label{f11sfig5}
    \end{subfigure}%
    \begin{subfigure}{.333\textwidth}
    \centering
     \includegraphics[width=0.8\linewidth]{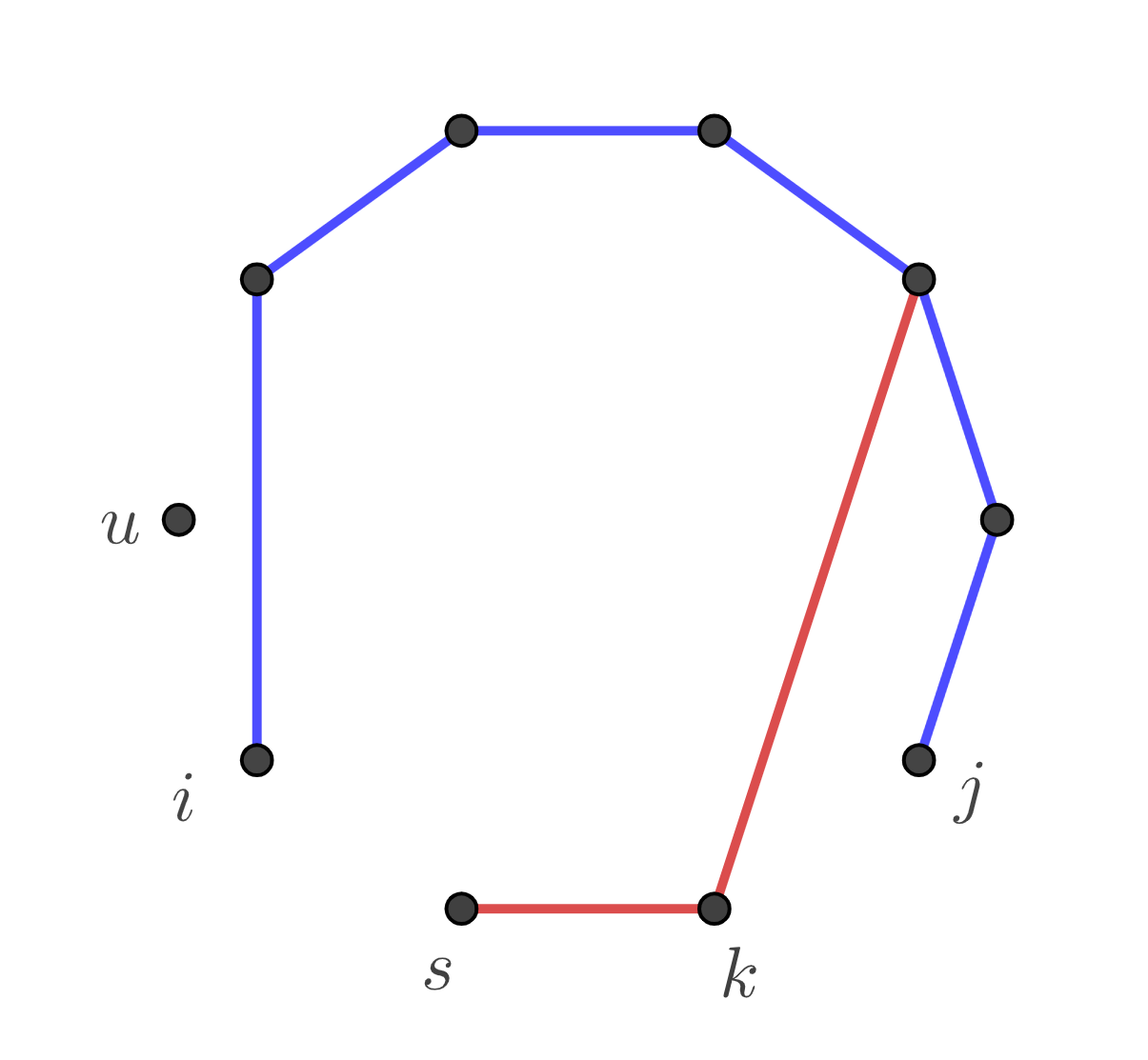}
     \caption{}
     \label{f11sfig6}
    \end{subfigure}%
    \captionsetup{justification=centering}
    \caption{Case 2: different possibilities for the graph when it is Blue's turn to play and his graph consists of a blue path on $n-3$ edges and three isolated vertices.}
    \label{f11}
    \end{figure}

    In his last two moves Blue claims the edges that complete the Hamiltonian cycle on the vertex set $V$.

\medskip
    Now we will show that when it is Blue's turn to play Case 2, Blue can follow his strategy and win.

    It is clear that Blue can claim an edge in his second move.
    If in his third move the edge $vy$ has been already claimed, that means that Red's graph has a red triangle and he will lose in his following move, so Blue can skip this move and claim the next edge and win.
    Otherwise, Blue claims the edge $vy$ and his graph at this moment consists of a $P_4$ and isolated vertices.

     \begin{cl}\label{cl41}
        If Blue's graph consists of a path disjoint from $u$ and more than three isolated vertices, Blue can extend his path by adding one of the isolated vertices from $V \backslash \{u\}$.
     \end{cl}
     \begin{proof}
     Let as denote by $P$ that blue path. There are at least three vertices from $V \backslash \{u\}$ that are not in $P$. Suppose that there is no edge such that Blue can extend his path with vertex from $V \backslash \{u\}$. That means that both ends of $P$ are incident with at least three red edges and that leads to a contradiction because there are two vertices of red degree three.
     \end{proof}

     According to Claim \ref{cl41} Blue can follow his strategy while there are more than two vertices in $V \backslash \{u\}$ that are not in the blue path. After that, one of the following cases happens:

     Case 2.1 \, It is obvious that the edge $sk$ cannot be red as otherwise Red would have a vertex of degree three. Therefore, Blue can claim his following two edges and make a path on $n-1$ vertices.

     Case 2.2 \, We can have two different options as depicted in Figure \ref{f11sfig3} and \ref{f11sfig4}.
     For the first one, obviously, Blue can claim two of the edges $ik,js,sk$ and make a path on $n-1$ vertices.
     For the second one, if $sk$ has become red Blue can take the last edge from $\{js, jk\}$ and win because Red has made a triangle. Otherwise, Blue claims the edge $sk$ and makes a path on $n-1$ vertices.

     Case 2.3 \, It is evident that Blue can claim his following two moves and make a path on $n-1$ vertices.

     Case 2.4 \, In case that $sk$ is free it is obvious.\\ Otherwise if the edge $sk$ is red, it cannot be an isolated edge in the Red's graph because his graph is connected. Therefore at least one of the vertices $\{s,k\}$ has to have degree two in Red's graph. Now, it is clear that Blue can follow his strategy.

     Case 2.5 \, In this case each of the edges $\{is, ik, js, jk\}$ is red and Red has a $C_4$ in his graph. Therefore, Blue can take the edge $iu$ and win.

     If Blue has not already won, at this moment his graph consists of a path on $n-1$ vertices and one isolated vertex $u$. Both edges that connect the vertex $u$ with ends of the blue path are free and Red cannot claim them. Therefore, Blue can claim these two edges in the following two moves and create the Hamiltonian cycle on the vertex set $V$ and win the game. \hfill $\Box$

\bigskip
{\bf Proof of Theorem \ref{TH003}:}
 First we will describe a strategy for Blue and then we will show that he can follow it and thus win.
 In the beginning, we have a graph $G$ with $n$ isolated vertices. After Red claims an edge, let us denote it by $uv$, Blue claims an edge that is not adjacent to the red one, let us denote it by $rt$. Because they have to play on connected graphs, in the following move Red, up to isomorphism, has two options to choose the following edge $e=xy$, and these will be our two cases.

\medskip
{\bf Case 1.} Vertex $x$ is red and $y$ is black, w.l.o.g. $x=u$.

 The following move of Blue is the edge $ut$, see Figure~\ref{f12sfig1}.

  \begin{figure}[ht]
    \begin{subfigure}{.5\textwidth}
    \centering
     \includegraphics[width=0.6\linewidth]{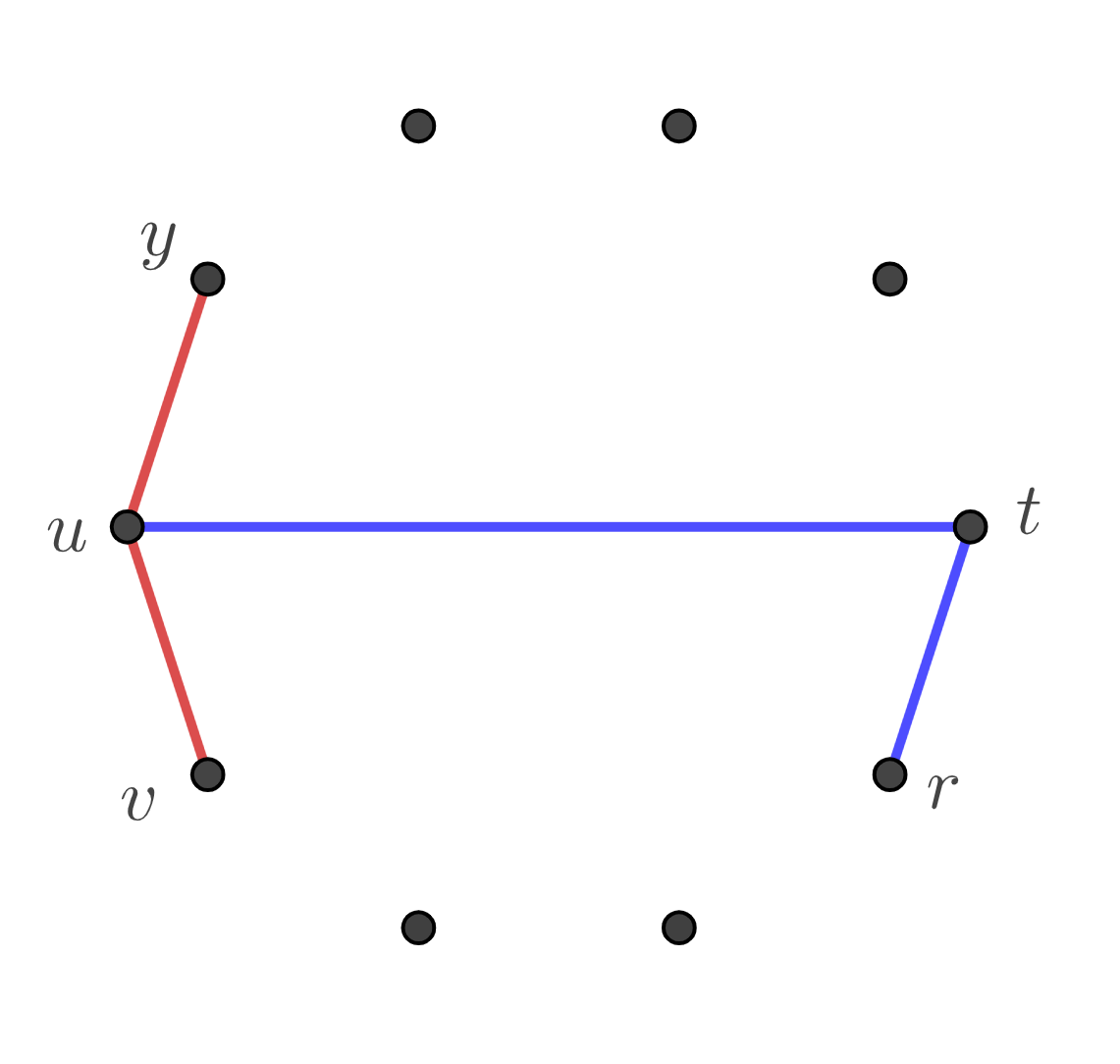}
     \caption{}
     \label{f12sfig1}
    \end{subfigure}%
    \begin{subfigure}{.5\textwidth}
     \centering
     \includegraphics[width=0.6\linewidth]{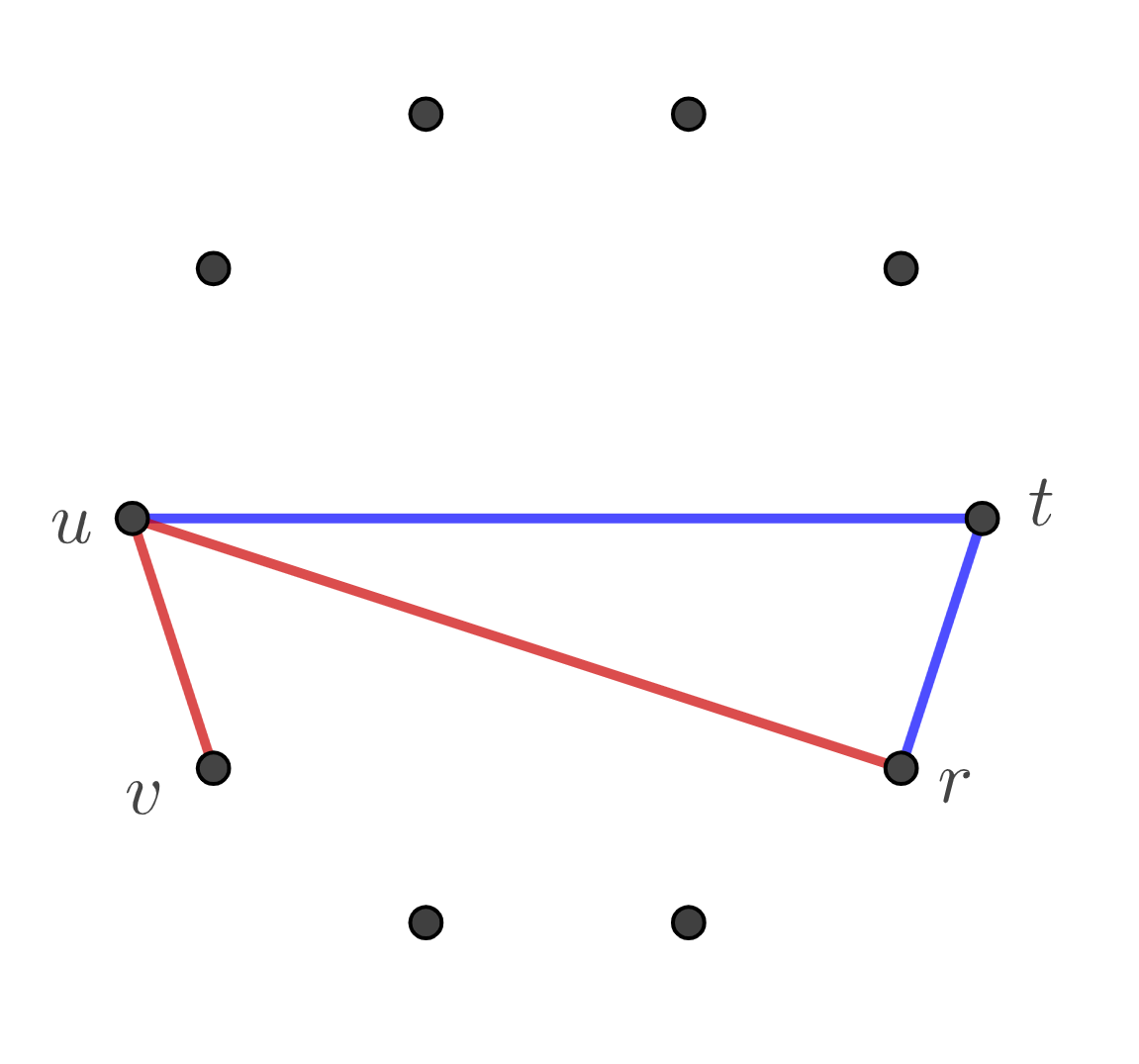}
    \caption{}
    \label{f12sfig2}
    \end{subfigure}
    \caption{The graph after the second move of Blue: (a) Case 1.  (b) Case 2.}
    \label{}
    \end{figure}
 Until the end of the game, Blue will star-add vertices that are not blue to the $t$-star.

\medskip
 \noindent{\bf Case 2.} Vertex $x$ is red and $y$ is blue,  w.l.o.g. $x=u$, $y=r$.

The following move of Blue is the edge $ut$, see Figure \ref{f12sfig2}.
 Until the end of the game, Blue will star-add vertices that are not blue to the $t$-star.

Now we will prove that Blue can follow his strategy and thus wins the game. In both cases it is evident that Blue can play first two moves and after that the graph of the game consists of two stars on three vertices (one red $u$-star and one blue $t$-star), and isolated vertices.

Note that Red cannot claim any edge incident with leaves of the $u$-star, not even close a red triangle, because his graph has to stay connected during the game, and if he claims a cycle he will inevitably lose the game. Therefore, only allowed moves for him are to star-add more vertices to the $u$-star. It is clear that Red cannot colour the vertex $t$ in red, hence Blue can claim each of $n-1$ edges of $t$-star. At the end of the game Red will claim at most $n-2$ edges, and then he has to make a $P_4$ in his following move. Therefore, Blue wins the game in $(n-1)$ rounds.
\hfill \qedsymbol

\vspace{0.2cm}

We proved Theorem \ref{TH003} giving an explicit strategy for Blue. It is straightforward to check that we can use a similar argument of strategy stealing as in Case 1 and Case 2 in the proof of Theorem \ref{TH001}, assuming that the graphs of both players stay connected throughout the game.

\bigskip

{\bf Proof of Theorem \ref{TH005}:}
As each player maintains his graph connected, it has to be a tree. Therefore, the game can last for at most $n-1$ rounds. If after the $(n-1)$-st round Blue's graph is a tree, than Red will lose in his following move.

We will show that Blue has a winning strategy. In the beginning, we have a graph $G$ with $n$ isolated vertices. After Red claims an edge, let us denote it by $uv$, Blue claims an edge that is not adjacent to the red one, let us denote it by $rt$. In the following move Red has two options, up to isomorphism, for choosing an edge $e=xy$, and those two moves will make our two cases.

\medskip
 {\bf Case 1.} Vertex $x$ is red and $y$ is blue, w.l.o.g. $x=u$ and $y=t$.\\
 We apply the strategy stealing argument, assuming that after his second move Red has a winning strategy $S$.

The graph of the game consists of one $P_4$, with 2 adjacent red edges and one blue edge, and isolated vertices, see Figure \ref{f13sfig1}. Before his next move, Blue imagines that he has already claimed the edge $vt$ and that Red has not claimed the edge $ut$, see Figure \ref{f13sfig2}.
Note that the edge $vt$ will remain free throughout the game, as otherwise Red would claim a triangle.

  \begin{figure}[ht]
    \begin{subfigure}{.5\textwidth}
    \centering
     \includegraphics[width=0.6\linewidth]{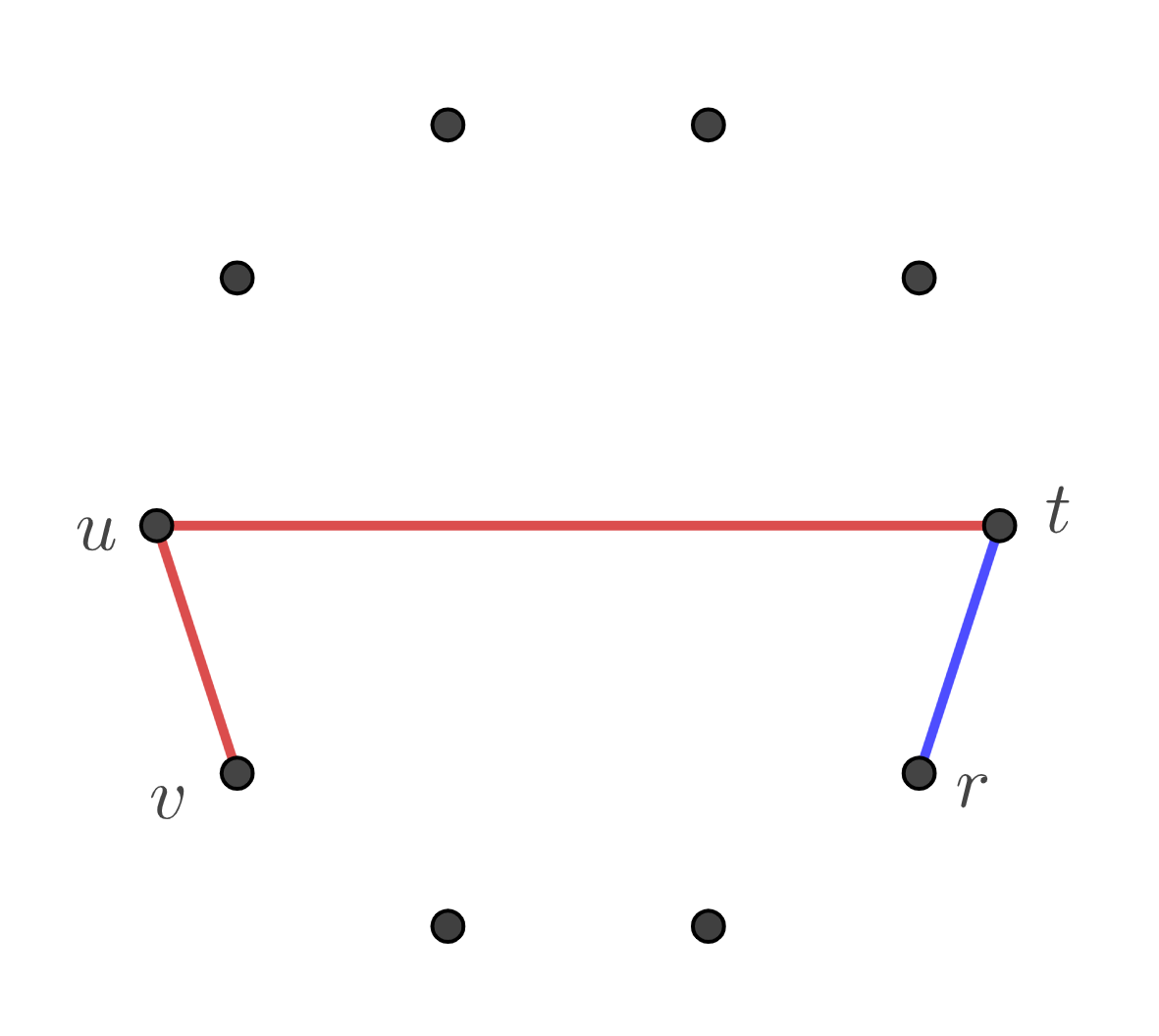}
     \caption{}
     \label{f13sfig1}
    \end{subfigure}%
    \begin{subfigure}{.5\textwidth}
     \centering
     \includegraphics[width=0.6\linewidth]{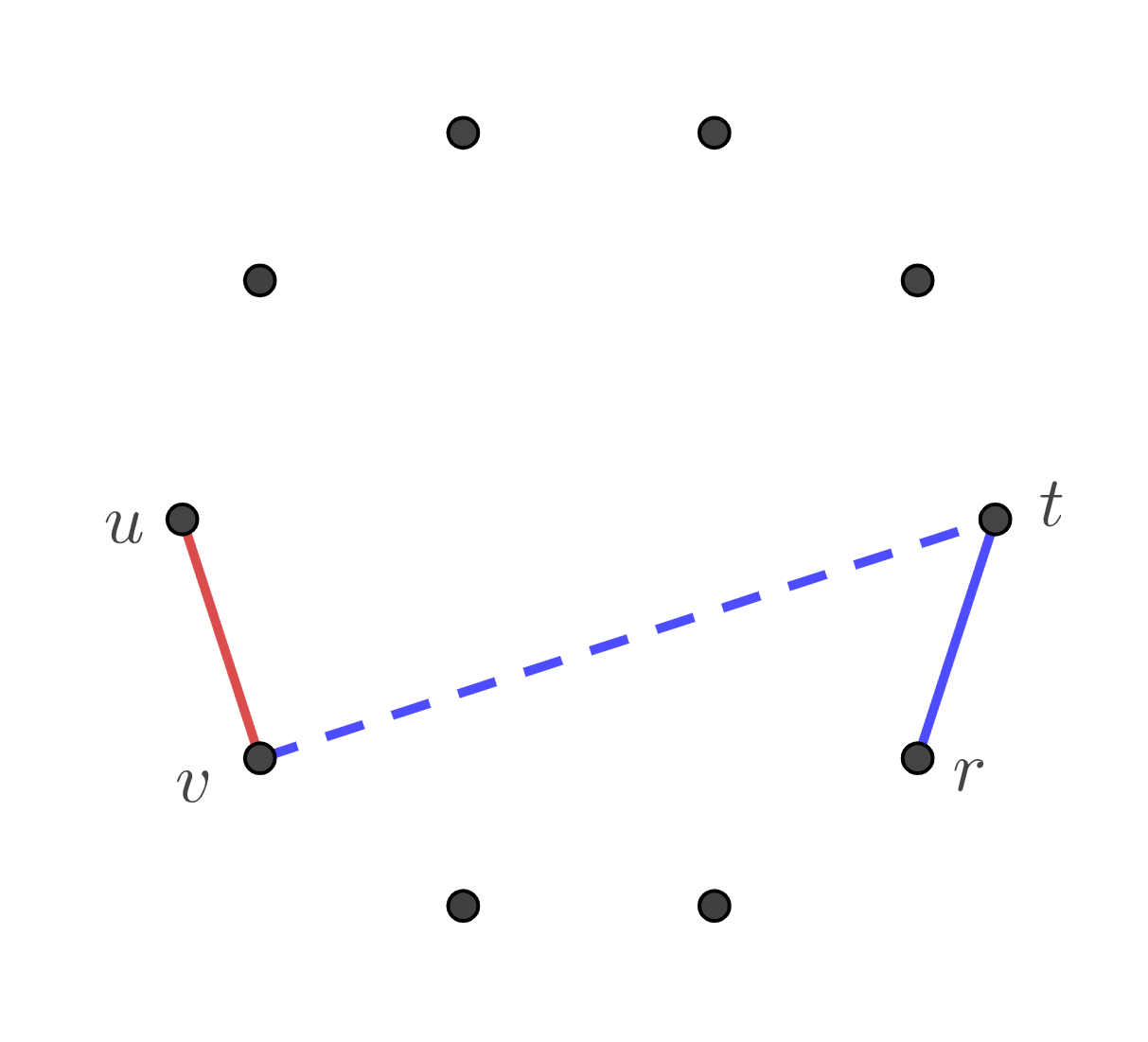}
    \caption{}
    \label{f13sfig2}
    \end{subfigure}
    \captionsetup{justification=centering}
    \caption{Case 1: (a) the graph before the second move of Blue. (b) The imagined graph before the second move of Red.}
    \label{f13}
    \end{figure}

 The imagined graph is isomorphic to the graph where the roles of players are swapped. Blue imagines that he is the first player and that Red claims the edge $ut$ as his second move. From now on, Blue responds as advised by the winning strategy $S$, and wins the game, a contradiction.

\medskip
 \noindent{\bf Case 2.} Vertex $x$ is red and $y$ is black, w.l.o.g. $x=u$.

 After Red claims the edge $uy$ it is Blue's turn, so he claims the edge $tu$, see Figure \ref{f9sfig1}. Then, no matter what Red plays, Blue claims an edge incident with $t$ and one black vertex. Then repeats that move for as long as possible. When there are no more black vertices, Blue will claim the edge incident with the pure red vertex of maximum degree and a blue vertex. He will continue doing that until all vertices in the graph are blue.

 Obviously, if Blue can follow his strategy, his graph will be a tree on $n-1$ edges and he will win. It remains to show that he can follow his strategy.

 While there is at least one black vertex, it is evident that Blue can make it adjacent to $t$ by claiming the edge incident to a black vertex and $t$.
 When there are no more black vertices in the graph, we will denote the pure red vertex of maximum degree by $m$, and by $i$ the number of edges in the Blue's graph. Note that at this moment Blue's graph consists of a $t$-star with $i \geq \Big\lfloor\frac{n-5}{2}\Big\rfloor+2$ edges, and isolated vertices.

 Assume for a contradiction that there is no free edge between $m$ and any blue vertex. Note that $m \neq u$ because $u$ is blue. That means that $m$ is adjacent to each of the blue vertices in the Red's graph. Therefore, Red must have a star with $i+1$ edges in his graph and at least one more edge from the beginning of the game. Hence, Red would have two edges more than Blue, a contradiction.
\hfill \qedsymbol


\begin{thebibliography}{9}
\bibitem{TicTacToe}
J. Beck, \emph{Combinatorial Games: Tic-Tac-Toe Theory}, Cambridge
University Press, 2008.

\bibitem{BekerStarAA}
A. Beker, \emph{The star avoidance game}, Australasian journal of combinatorics, {\bf 77} (2020), 398-405.


\bibitem{ChE}
V. Chv\' atal, P. Erd\H{o}s, \emph{Biased positional games}, Annals of Discrete Mathematics, {\bf 2} (1978), 221-229.

\bibitem{KrivelevichWalkerBreaker}
L. Espig, A. Frieze, W. Pegden and M. Krivelevich, \emph{Walker–Breaker
Games}, SIAM Journal on Discrete Mathematics, {\bf 29} (2015), 1476–1485.

\bibitem{Hefetzstrongfast}
A. Ferber and D. Hefetz, \emph{Winning strong games through fast strategies for weak games}, The Electronic Journal of Combinatorics, {\bf 18} (2011), P144.

\bibitem{Hefetzstrongkconnectivity}
A. Ferber and D. Hefetz, \emph{Weak and strong k-connectivity game}, European Journal of Combinatorics, {\bf 35} (2014), 169-183 (special issue of
EuroComb11).

\bibitem{MimaJovanaWMWB}
J. Forcan and M. Mikala\v cki, \emph{On the WalkerMaker–WalkerBreaker
games}, Discrete Applied Mathematics, {\bf 279} (2020), 69–79.

\bibitem{Harary}
F. Harary, \emph{Achievement and avoidance games for graphs}, Ann. Discrete Math., {\bf 62} (1982), 111–119.

\bibitem{HararySlanyVerbitsky}
F. Harary, W. Slany, O. Verbitsky, \emph{A Symmetric Strategy
in Graph Avoidance Games}, More games of no chance, {\bf 42} (2002), 369-381.

\bibitem{StojakovicTheRulesOfTheGame}
D. Hefetz, M. Krivelevich, M. Stojakovi\' c and T. Szab\' o, \emph{Avoider–Enforcer: The rules of the game}, Journal of Combinatorial Theory
Series A, {\bf 117}  (2010), 152–163.

\bibitem{BookStojakovic}
D. Hefetz, M. Krivelevich, M. Stojakovic and T. Szabó, \emph{Positional Games}, Oberwolfach Seminars, Vol. 44, Birkhäuser Basel, Springer, 2014.

\bibitem{TransitiveAvoidance}
 J. R. Johnson, I. Leader and M. Walters, \emph{Transitive Avoidance Games}, Electronic Journal of Combinatorics, {\bf24} (2017), P1.61.

\bibitem{Prim}
A. London and A. Pluhár, \emph{Spanning tree game as Prim would have played}, Acta Cybernetica, {\bf 23} (2018), 921-927.

\bibitem{Malekshahian}
A. Malekshahian, \emph{Strategy Stealing in Triangle Avoidance Games}, arXiv preprint (2020), \url{https://arxiv.org/abs/2001.10116}.

\bibitem{SIM74}
E. Mead, A. Rosa and C. Huang, \emph{The Game of Sim: A winning strategy for the second player}, Mathematics Magazine, {\bf 47} (1974), 243-247.

\bibitem{Simmons69}
G. J. Simmons, \emph{The game of SIM}, J. Recreational Mathematics, {\bf 2} (1969), 66.

\bibitem{Slany99}
W. Slany, \emph{Graph Ramsey Games}, Available:\\
\url{https://eccc.weizmann.ac.il/eccc-reports/1999/TR99-047/index.html}.

\bibitem{Sim20}
W. Wrzos-Kaminska, \emph{A simpler winning strategy for Sim}, arXiv preprint (2020),\\ \url{https://arxiv.org/abs/2001.04024}.



















\end{thebibliography}
\end{document}